\documentclass[10pt]{article}
\usepackage{arxiv}

\usepackage{cite}
\usepackage{graphicx}
\usepackage{textcomp}
\usepackage{xcolor}
\usepackage{subcaption}
\usepackage{amsmath}
\usepackage{amsfonts}
\usepackage{mathrsfs}
\usepackage{amssymb}
\usepackage{amsthm}
\usepackage{dsfont}
\usepackage{times}
\usepackage{multirow}
\usepackage[none]{hyphenat}
\usepackage{float}
\usepackage{iftex}
\usepackage{algorithm}
\usepackage{algpseudocode}
\usepackage{booktabs}
\usepackage{caption}
\usepackage{tikz}
\usepackage{tikz-qtree}
\usepackage{array}
\usepackage{adjustbox}
\usepackage[acronym]{glossaries}
\usepackage{diagbox}
\usepackage{hyperref}
\usepackage{cleveref}

\hypersetup{
  colorlinks=true,
  linkcolor=blue,
  citecolor=blue,
  urlcolor=blue
}

\def\BibTeX{{\rm B\kern-.05em{\sc i\kern-.025em b}\kern-.08em
    T\kern-.1667em\lower.7ex\hbox{E}\kern-.125emX}}

\newtheorem{theorem}{Theorem}
\newtheorem{corollary}{Corollary}
\newtheorem{lemma}{Lemma}

\newtheorem{assumption*}{Assumption}
\newtheorem{stdassumption*}{Standing Assumption}

\newtheorem{definition}{Definition}
\newtheorem{definition*}{Definition}

\newacronym{tdth}{TD-TH}{time domain thresholding}
\newacronym{tfdth}{TFD-TH}{time-frequency domain thresholding}
\newacronym{fmcw}{FMCW}{frequency modulated continuous wave}
\newacronym{pd}{PD}{probability of detection}
\newacronym{sinr}{SINR}{signal-to-interference-plus-noise ratio}
\newacronym{pri}{PRI}{pulse repetition interval}
\newacronym{adc}{ADC}{analog-to-digital converter}
\newacronym{fft}{FFT}{fast Fourier transform}
\newacronym{lrr}{LRR}{long-range radar}
\newacronym{srr}{SRR}{short-range radar}
\newacronym{cfar}{CFAR}{constant false alarm rate}
\newacronym{rfe}{RFE}{radar front end}
\newacronym{lpf}{LPF}{low-pass filter}
\newacronym{rd}{RD}{range-Doppler}
\newacronym{cdf}{CDF}{cumulative distribution function}
\newacronym{adas}{ADAS}{advanced driver assistance system}
\newacronym{ad}{AD}{autonomous driving}
\newacronym{stft}{STFT}{short-time Fourier transform}
\newacronym{aeb}{AEB}{automatic emergency braking}
\newacronym{vru}{VRU}{vulnerable road user}
\newacronym{acc}{ACC}{adaptive cruise control}
\newacronym{tx}{TX}{transmitter}
\newacronym{rx}{RX}{receiver}
\newacronym{lfm}{LFM}{linear frequency modulated}
\newacronym{cpi}{CPIs}{coherent processing intervals}
\newacronym{vtx}{V2X}{vehicle-to-everything}
\newacronym{RadCom}{RadCom}{radar-communication}
\newacronym{snr}{SNR}{Signal-to-noise Ratio}

\newcommand{\C}{\mathbb{C}}
\newcommand{\Tpri}{T_{\text{PRI}}}
\newcommand{\Ical}{ \mathcal{I} }

\renewcommand{\headeright}{}
\renewcommand{\undertitle}{}
\renewcommand{\shorttitle}{Decentralized No-Regret Frequency-Time Scheduling for FMCW Radar Interference Avoidance}

\title{Decentralized No-Regret Frequency-Time Scheduling for FMCW Radar Interference Avoidance}

\author{
Yunian Pan\\
Department of Electrical and Computer Engineering\\
New York University\\
New York, NY, USA\\
\texttt{yp1170@nyu.edu}
\And
Jun Li\\
NXP Semiconductors\\
San Jose, CA, USA\\
\texttt{jun.li\_5@nxp.com}
\And
Lifan Xu\\
Department of Electrical and Computer Engineering\\
The University of Alabama\\
Tuscaloosa, AL, USA\\
\texttt{lxu36@crimson.ua.edu}
\AND
Shunqiao Sun\\
Department of Electrical and Computer Engineering\\
The University of Alabama\\
Tuscaloosa, AL, USA\\
\texttt{shunqiao.sun@ua.edu}
\And
Quanyan Zhu\\
Department of Electrical and Computer Engineering\\
New York University\\
New York, NY, USA\\
\texttt{qz494@nyu.edu}
}

\date{}

\begin{document}

\maketitle

\begin{abstract}
    
Automotive frequency modulated continuous wave (FMCW) radars are indispensable to modern advanced driver assistance systems (ADAS) and autonomous-driving systems, but their increasing density has intensified the risk of mutual interference. Existing mitigation techniques, including reactive receiver-side suppression, proactive waveform design, and cooperative scheduling, often face limitations in scalability, reliance on side-channel communication, or degradation of range–Doppler resolution. Building on our earlier work on decentralized frequency-domain no-regret hopping, this paper introduces a unified time–frequency game-theoretic framework that enables radars to adapt across both spectral and temporal resources. We formulate the interference-avoidance problem as a repeated anti-coordination game, in which each radar autonomously updates a mixed strategy over frequency subbands and chirp-level time offsets using regret-minimization dynamics. We show that the proposed Time–Frequency No-Regret Hopping algorithm achieves vanishing external and swap regret, and that the induced empirical play converges to an $\varepsilon$-coarse correlated equilibrium or a correlated equilibrium. Theoretical analysis provides regret bounds in the joint domain, revealing how temporal adaptation implicitly regularizes frequency selection and enhances robustness against asynchronous interference. Numerical experiments with multi-radar scenarios demonstrate substantial improvements in signal-to-interference-plus-noise ratio (SINR), collision rate, and range–Doppler quality compared with time-frequency random hopping and centralized Nash-based benchmarks. 
% These results show that extending no-regret learning to joint time–frequency adaptation provides a scalable and communication-free solution for reliable radar operation in dense interference environments.

\end{abstract}

\glsresetall

\section{ Introduction}

%% Broad Context & Importance
High-resolution \gls{fmcw} radars are central to \gls{adas} and \gls{ad}, providing accurate range, velocity, and angle estimates under diverse environmental conditions \cite{Sun_SPM_2020, Sun_JSTSP_4D_2021, Automotive_Radar_Challenges_ICASSP_2025,Lifan_Co_Chirps_TAES_2023,Ruxin_Radar_Imaging_TAES_2023,Markel_book_2022,BLRC,tao23cola}. \gls{fmcw} radars enable hazard detection and autonomous decision-making, serving as a cornerstone of vehicle safety and autonomy, and are thus essential for both current and future vehicles.

\subsection{Motivation}

% Bandwidth is a resource allocation
Modern automotive \gls{fmcw} radars predominantly operate in the 76-81 GHz bands \cite{hakobyan2019high}, with some systems, such as Forvia/Hella’s 3rd and 4th generations, also utilizing 24 GHz.
While the 24 GHz ISM band offers limited bandwidth (200–250 MHz depending on region), the 76–81 GHz band supports both long- and short-range sensing, making it central to \gls{adas}. 
As radar adoption surges across automotive platforms, the risk of mutual interference has grown acute.  Without effective interference mitigation, \gls{sinr} deteriorates, compromising detection performance in safety-critical functions such as \gls{aeb}, \gls{vru} protection, and \gls{acc}.
The importance of interference mitigation has been recognized across industry and regulatory bodies through initiatives such as the EU MOSARIM program, the NHTSA Radar Congestion Study \cite{buller2018radar}, and Germany’s IMIKO-Radar project \cite{ossowska2021imiko}. A broad line of research has since examined interference mitigation using orthogonal signaling, reactive receiver-side suppression, proactive waveform design, protocol-based scheduling, and more recently, \gls{vtx}-assisted coordination. Despite these developments, ensuring reliable radar operation under dense and heterogeneous deployments remains challenging. The limitations of centralized or communication-dependent approaches highlight the need for adaptive and fully decentralized interference-avoidance mechanisms.
% Although largely invisible to consumers, interference has become a pressing concern for the automotive and radar communities, as evidenced by initiatives such as the EU’s MOSARIM project, the U.S. NHTSA’s Radar Congestion Study \cite{buller2018radar}, and Germany’s IMIKO-Radar project \cite{ossowska2021imiko}. These efforts, alongside a broad spectrum of academic research, have explored mitigation through orthogonal signaling, reactive countermeasures, proactive scheduling, and, more recently, \gls{vtx}-enabled coordination and novel waveform designs. Despite this progress, the challenge of ensuring reliable radar operation under dense deployment remains unresolved, underscoring the need for adaptive and decentralized interference avoidance strategies.

%% Rising Interference Challenge
% To address this issue, a variety of interference mitigation techniques have been proposed \cite{li2024performance,jun2022radar,jeroen19,Mun18,Mun20,jihwan2024intf,xinyi2024}, aiming to enhance the functionality and reliability of automotive radar systems. The degree of improvement required in Signal-to-Interference-plus-Noise Ratio (SINR) after mitigation depends on the specific application, ranging from modest enhancements of a few decibels to complete elimination of interference. For instance, in critical applications such as radar-based Vulnerable Road User (VRU) protection, any sudden system failure caused by interference is unacceptable, necessitating robust strategies to ensure complete interference avoidance and operational safety.

\subsection{Existing Interference Mitigation Approaches}

A broad spectrum of interference mitigation strategies has been investigated to sustain the reliable operation of automotive radars under dense deployment \cite{Sun_SPM_2020,Sun_JSTSP_4D_2021}. These techniques can be broadly classified into three sub-categories.

\subsubsection{\textbf{Receiver-side suppression methods}} focus on detecting and suppressing interference after reception. Common strategies include time-domain blanking or clipping \cite{li2024performance}, time-frequency thresholding \cite{jun2022radar,li2024time_frequency_dual_domain,jihwan2024intf}, and adaptive filtering \cite{jeroen19}, which selectively remove contaminated portions of the received signal. More recent work leverages deep learning models such as score-based generative networks, recurrent neural networks, and attention mechanisms to identify and reconstruct interference-free signals \cite{xinyi2024,xinyi2025,Mun20}. Although these approaches can be effective in moderate interference conditions, their performance degrades when the number of interferers increases. In dense scenarios, interference often overlaps with the target echo in both the time domain and the time–frequency domain, making it fundamentally difficult to reliably separate and recover the original signal energy and thereby fully restore the \gls{snr}. Moreover, some deep learning–based solutions may introduce additional latency or computational cost.

\subsubsection{\textbf{Transmitter-side (proactive) methods}} aim to prevent interference before it occurs by diversifying radar waveforms. Classical techniques exploit orthogonality in time, frequency, or code domains, while more advanced designs employ nonlinear frequency hopping or multiband chirps to reduce collision probability while preserving range resolution \cite{Lifan_Co_Chirps_TAES_2023,Sun_JSTSP_4D_2021,Yimin_Partioning_SAM_2024,stettiner2023fmcw}. These approaches can improve the received \gls{sinr}, but they also exhibit several limitations. Without coordination among radars, waveform diversification remains vulnerable to random overlap, especially in dense traffic where the available time–frequency space is limited. In addition, many waveform modifications introduce inherent trade-offs with sensing performance, including reduced range or velocity resolution, increased sidelobes, or degraded Doppler coherence. The resulting designs are often difficult to adapt in real time to rapidly changing interference conditions, and their effectiveness diminishes in heterogeneous environments where different radar platforms operate with different bandwidths, slopes, and timing configurations. 

% Such approaches directly improve the \gls{sinr}, but without structured coordination they remain vulnerable to random overlap and can limit adaptability in dynamic environments.

\subsubsection{\textbf{Cooperative and protocol-based methods}} leverage communication between vehicles or radars to coordinate spectrum access. Examples include MAC-style scheduling \cite{Google_Automotive_Interference_Mitigation_2016}, radar networking via LTE/5G, and \gls{RadCom} systems enabled by \gls{vtx} links \cite{Jin_FMCW_Interference_JSTSP_2021}. These methods can, in principle, guarantee interference-free operation, but they rely on robust side channels and synchronized participation, which may not be feasible in heterogeneous or adversarial settings.

%A detailed comparison is listed in Table~\ref{tab:comparison}.
% Despite these advances, ensuring reliable radar operation in congested traffic remains unresolved. Reactive methods alone cannot fully eliminate interference, proactive waveform design often trades resolution for robustness, and cooperative protocols face deployment and security barriers \cite{pan2025game}. Even proactive techniques such as frequency hopping \cite{stettiner2023fmcw} and beamforming \cite{Jin_Interference_TVT_2024} suffer from random parameter assignment or the need for coordination. This gap motivates the development of adaptive and decentralized strategies that operate without side-channel communication, scale gracefully with radar density, and balance the trade-off between interference avoidance and range resolution.

% TALK ABOUT Why GAME THEORY IS IMPORTANT
% 
\subsection{No-Regret Hopping As a Hybrid Approach}

To address the aforementioned issues, we introduced a game-theoretic framework for frequency scheduling in \cite{pan2025game}, where each radar is modeled as a strategic player that selects frequency-hopping schemes to optimize its own performance under mutual interference. In this formulation, the utility of each player is defined by its resulting \gls{sinr}, and the goal is to learn adaptive hopping strategies that mitigate interference while preserving range resolution.  

Within noncooperative game theory, two canonical solution concepts are often considered: 1) \textit{Nash Equilibrium} (NE) \cite{bacsar1998dynamic}, where no player has an incentive to unilaterally deviate from its chosen strategy. 2) \textit{Coarse Correlated Equilibrium} (CCE) \cite{roughgarden2010algorithmic}, a broader class of equilibria defined as joint distributions over strategies such that players cannot improve their expected utility by deviating from signaled recommendations.  
We demonstrated that due to several drawbacks of NE in radar scenarios, namely equilibrium selection, price of anarchy, and computational infeasibility (PPAD-hard \cite{thecomplexityofne}), learning a CCE turns out to be more computationally feasible and economical \cite{pan2023resilience,pan2024variationalinterpretationmirrorplay,tao22confluence}.  
Building on this insight, the proposed \textit{No-Regret Hopping}, a decentralized algorithm that leverages regret minimization to approach a CCE, demonstrated several advantages: it requires no inter-radar communication, scales naturally to dense deployments, and empirically outperformed both the centralized \textit{Nash Hopping} approach \cite{9682998} and uniformly random hopping. 

While these results established No-Regret Hopping as a robust frequency-domain solution, the framework in \cite{pan2025game} remained restricted to probabilistic frequency allocation across chirps and therefore did not exploit the temporal dimension available between consecutive linear chirps. As illustrated in Fig.~\ref{fig:timeshiftillustration}, limiting adaptation to frequency shifting alone overlooks the additional temporal resource, and independent subband sampling across chirps inevitably induces a trade-off between range resolution and interference avoidance. To overcome this limitation, the present work extends the no-regret learning framework into the \textit{time-shifting domain}, enabling joint adaptation over both frequency and time. This extension enables a broader utilization of the allocable frequency-time resource. 
\begin{figure}
    \centering
    \includegraphics[width=0.9\linewidth]{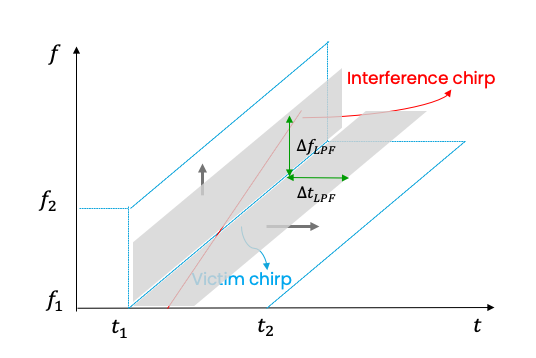}
    \caption{Illustration of interference mitigation through frequency and time shifting. When interference occurs, two potential adaptation directions are available (gray arrows). }
    \label{fig:timeshiftillustration}
\end{figure}

%%%%% Contribution and outline of the paper %%%%%%% 

%%%%%%%%%%%%%%%%%%%%%%%%%%%%%%%%%%%%%%%%%%%%%%%%%%%%%%%%%%%%%%%%%%%%%%%%%%%%%

\subsection{Related Work and Contributions}
Several prior studies \cite{Lifan_Co_Chirps_TAES_2023} have explored time-domain techniques to mitigate radar interference in distributed environments. One common approach involves introducing chirp-level time shifts within a radar frame to reduce the likelihood of simultaneous transmissions from multiple devices. For example, some works propose random or pseudo-random time dithering of chirp start times to decorrelate overlapping signals and suppress coherent interference. These methods are particularly effective in uncoordinated settings where centralized scheduling is not feasible. In addition, time-division randomization has been shown to improve detection performance by spreading potential collisions across the slow-time dimension, thereby reducing persistent interference patterns. While these techniques can reduce interference probability, they are often used in isolation and may not fully exploit the benefits of combining time and frequency diversity. Our work builds on these insights by integrating slow-time chirp shifting with frequency hopping, aiming to achieve a more robust and scalable interference avoidance strategy in dense radar networks.

\textbf{Time-shifting strategies.} While most prior research on interference mitigation has focused on frequency-domain approaches, a growing body of work recognizes the importance of temporal alignment between chirps. Early reactive methods employed time-domain blanking or delay-based suppression to mitigate overlapping interference \cite{li2024performance,jeroen19}, but these strategies do not proactively exploit the temporal degrees of freedom. More recent approaches have considered randomized pulse repetition intervals (PRIs) or asynchronous chirp scheduling \cite{jun2022radar,Jin_FMCW_Interference_JSTSP_2021}, which reduce collision probability but lack formal guarantees and often degrade range–Doppler resolution. To the best of our knowledge, no existing work integrates time-domain shifting into a learning-theoretic framework with provable convergence guarantees.

\textbf{Our contributions.} This paper extends the previous \textit{No-Regret Hopping} framework \cite{pan2025game} in the frequency-domain into the time-shifting domain, thereby unifying frequency and temporal adaptation for FMCW radar interference avoidance. The main contributions are as follows:
\begin{itemize}
    \item We reformulate the radar interference avoidance problem as a repeated game where radars adapt both their subband allocation and temporal chirp positioning, thus capturing the joint frequency–time trade-off.
    \item We develop a decentralized \textit{Time–Frequency No-Regret Hopping} algorithm that guarantees vanishing external regret, leading to convergence toward an $\varepsilon$-Correlated/Coarse correlated equilibrium (CCE/CE), providing a theoretical analysis of regret bounds in the joint domain and demonstrating how implicit temporal regularization improves robustness against asynchronous interference.
    \item We validate the proposed framework through simulations that show significant improvements in SINR and target detection quality compared to purely randomized frequency-hopping-time-shifting strategy, and benchmark predetermined Nash equilibrium strategies.
\end{itemize}

\textbf{Organization of the paper.} The remainder of the paper is structured as follows. Section~\ref{sec:radarsystemasgame} introduces the radar signal model and formulates the interference avoidance game with joint frequency–time resources. Section~\ref{sec:solconceptalgo} presents the proposed no-regret learning framework and provides theoretical analysis. Section~\ref{sec:experiment} reports numerical simulations and performance evaluations. Finally, Section~\ref{sec:conclusion} concludes the paper and discusses future directions, including heterogeneous radar deployments and real-time implementation challenges.

\section{Game  Formulation and Signal Model}
\label{sec:radarsystemasgame}

We consider a system comprising a set of automotive FMCW radars denoted by $\Ical = \{1, \ldots, I \}$, each operating in real time to detect a single target object. 
Every radar transmits \gls{lfm} chirps over a total bandwidth $B_{\text{tot}}$ centered at frequency $f_c$.

The radar operations span $\mathcal{T}$ \gls{cpi}, indexed by $\tau \in \{1, \ldots, \mathcal{T}\}$, each CPI consists of duration $T$, within which a radar $i \in \Ical$ transmits $K^i$ discrete chirps sequentially. The \gls{pri} for each radar $i$ is given by $\Tpri^{i} = \frac{T}{K^{i}}$,
which consists of an active period $T^{i}_a$ and an idle time $T^{i}_d$, such that
\begin{equation*}
    \Tpri^{i} = T^{i}_a + T^{i}_d, \quad \text{with} \quad T^{i}_d > 0.
\end{equation*}
During the active interval $T^{i}_a$, the radar emits a chirp that sweeps a subband of bandwidth $B^{i}$, yielding a chirp slope of $\alpha^{i} = \frac{B^i}{T^{i}_a}$.
% In this work, we assume homogeneous radar configurations for notational simplicity, i.e., $K^i = K$, $T^i_a = T_a$, and $T^i_d = T_d$ for all $i \in \Ical$. We discuss the case under non-homogeneous radar configuration in the Appendix.

We discretize the frequency and time resources as follows:
\begin{itemize}
    \item The total bandwidth $B_{\text{tot}}$ is partitioned into $A_f$ (potentially overlapped) subbands $\{f_1, \ldots, f_{A_f}\}$, where each subband has bandwidth $B^i$ and is defined by a starting frequency $f_a = f_c + (a-1)\Delta f$, for $a \in \{1, \ldots, A_f\}$ with $\Delta f$ denoting the subband frequency offset (or hopping).
    \item The idle duration $T_d$ is illustrated in \Cref{fig:idletime}. The idle time is divided into $A_t$ discrete time slots $\{t_1, \ldots, t_{A_t}\}$, where each represents the starting time of the active period $t_b = (b-1)\Delta t$, for $b \in \{1, \ldots, A_t\}$ with $\Delta t$ denoting the chirp time domain shifting.
\end{itemize}

\begin{figure}
    \centering
    \includegraphics[width=.9\linewidth]{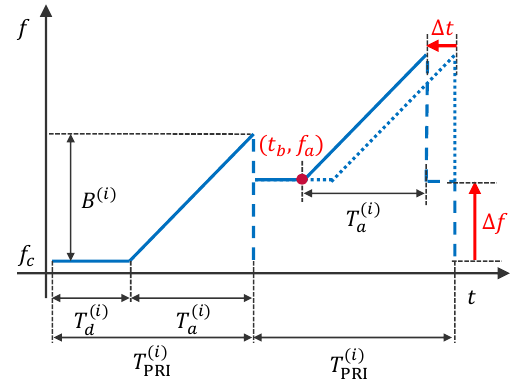}
    \caption{Joint time–frequency scheduling illustration. 
    Each chirp selects a subband start $f_a$ and a time offset $t_b$ within the idle segment $T_{d}^{(i)}$. 
    The subband frequency hopping and time shifting width are $\Delta f$ and $\Delta t$, forming the joint action $(t_b, f_a)$.}
    \label{fig:idletime}
\end{figure}

Hence, every chirp corresponds to a joint time–frequency action. The joint time–frequency action space is a Cartesian product
$
    \mathcal{A} = \{f_1, \ldots, f_{A_f}\} \times \{t_1, \ldots, t_{A_t}\},
$
so that radar $i$ selects one subband $f_a$ and one start time $t_b$ per chirp.

\subsection{Graphical Game Model}

To proactively avoid interference, the radars strategically select frequency hopping and time shifting sequences $(f^{i}_1, t^{i}_1, \ldots, f^{i}_{K^{i}}, t^{i}_{K^{i}})$ across chirps within each CPI. 
We denote the composition of the channel selection and active period selection strategies of radar $i$ by $\sigma^{i}(\cdot): \{1, \ldots, K^{i}\} \to \mathcal{A}$; for example, $\sigma^{i}(k) = (f, t)$ indicates that radar $i$ selects subband $[f, f+ B]$ at chirp $k$, sending one LFM pulse between $[t, t+T_a ]$. Let $\Sigma^{i}$ denote the strategy set of radar $i$. 
Note that $\Sigma^{i}$ does not necessarily exhaust all possible mappings; for instance, it may be defined as a set of periodic mappings with an offset $n$, that is,
$
\Sigma^{i} := \bigl\{ \sigma \mid \sigma^i(k) = \bigl( f_{(k+n) \bmod A_f},\; t_1 \bigr),\ n \in \mathbb{N} \bigr\},
$
% in which case the frequency hopping is periodically linear with respect to chirp indices, and the time shifting strategy is arbitrary. We illustrate such hopping strategy in a 4 radar case in \Cref{fig:4radarround-robin}, 
% \begin{figure}
%     \centering
%     \includegraphics[width=0.9\linewidth]{figs/4radarroundrobin}
%     \caption{The 4 radar round-robin case}
%     \label{fig:4radarround-robin}
% \end{figure}
Note that in general, the design of the strategy sets is not restricted to such configurations.

\begin{figure}
    \centering
    \includegraphics[width=.9\linewidth]{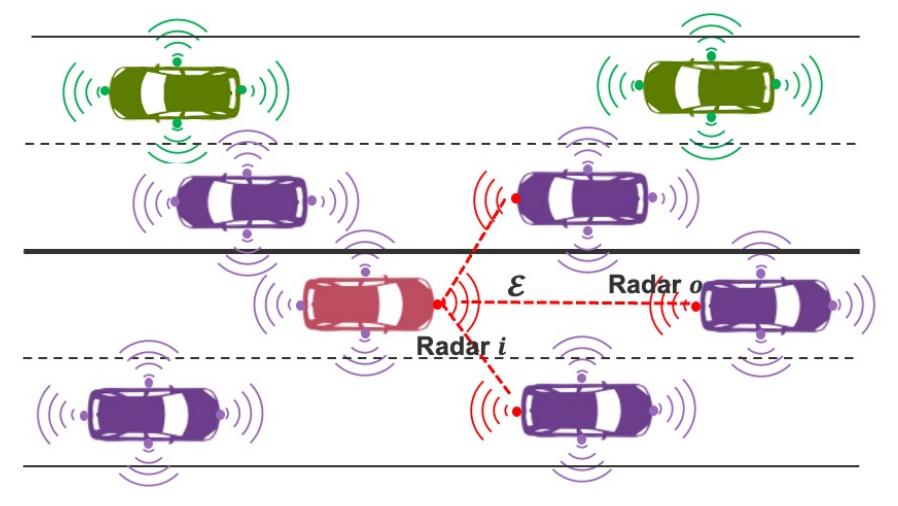}
    \caption{Typical automotive radar interference scenario. The red vehicle’s front radar receives interference from three radars on nearby purple vehicles, i.e., two front-facing and one rear-facing. Mutual interference links among these distributed radars form a graphical game representation.}
    \label{fig:fig_scenario}
\end{figure}

Then, the typical radar-to-radar interference scenario as shown in Fig.~\ref{fig:fig_scenario} can be modeled as a $\mathcal{T}$-stage repeated graphical game $\Gamma = \{ \Ical, \{ \Sigma^{i} \}_{i \in \Ical }, \mathcal{G}, \{ U^{i} \}_{i \in \Ical }\}$, where the stages $\mathcal{T}$ are CPIs; the radars, acting as players in $\Ical$, share a common strategy set $\Sigma$. 
The interference relationships among a set of radars are modeled by a weighted, directed graph 
\(\mathcal{G} = (\Ical, \mathcal{E}, \Delta)\), 
where each node \(i \in \Ical\) corresponds to radar \(i\), and each directed edge \((i \to o)\in \mathcal{E}\) indicates that radar \(i\) can interfere with radar \(o\). 
In addition, each directed edge $(i  \to  o)\in\mathcal{E}$ carries a delay $\Delta_{io}$ (seconds), the arrival-time offset of radar $i$’s waveform as measured in radar $o$’s local clock. We assume $\mathcal{G}$ is CPI-invariant for simplicity.  Therefore, \(\mathcal{G}\) fully specifies who interferes with whom and when. In this work, we assume $\mathcal{G}$ is stage independent for simplicity.
The (also stage-independent) utility functions $U^{i}: \prod_{i \in \Ical}\Sigma^{i} \to \mathbb{R}$ endow this game with the anti-coordination property, where the pulse transmission strategies $\sigma^{i}$ employed by players are from potentially different strategy sets $\Sigma^{i}$, but whenever two or more radars' transmission strategies agree, say, $\sigma^{i}(k) \approx \sigma^{o}(k)$ at chirp $k$ for some radars $i$ and $o$, the utilities of both radars decrease. Specifically, $(f^{i}_k, t^{i}_k)$ being close to $(f^{o}_k, t^{o}_k)$ degrades the detection performance of both radar $i$ and $o$ when $(i,o) \in \mathcal{E}$ and the arrival times $t^{i}_k$ and $t^{o}_k + \Delta_{io}$ are close enough.
In this paper, we give the following definition for players' utility functions, for each stage $\tau = 1, \ldots, \mathcal{T}$,
% SENSING GIVING YOU THE UTILITY
\begin{equation}\label{eq:utility}
    U^{i} (\sigma^{i}_{\tau}, \sigma^{{-i}}_{\tau}) =   u(\mathrm{SINR}^{i}_{\tau}) , \quad 
\end{equation}
where $\mathrm{SINR}^{i}_{\tau}$ represents the average obtained \gls{sinr} of CPI $\tau \in \{1, \ldots, \mathcal{T} \}$, at radar $i$'s receiver, determined by the joint strategies $\sigma^{i}_{\tau}, \sigma^{{-i}}_{\tau}$ for CPI $\tau$; $u(\cdot)$\footnote{In practice we use a saturating map $g(s)=\frac{s^\beta}{s^\beta+S_0^\beta}$ on linear SINR, with $\beta  \in  (0,\infty)$ and $S_0$ denoting the SINR threshold at which the utility equals $0.5$, i.e., $u(S_0)=0.5$; this facilitates algorithmic stability.} is a strictly increasing differentiable function that maps the \gls{sinr} to a scalar in $(0,1)$.
The scalar $\bar U^i_\tau$ is appended to the utility buffer $\mathcal U$ and fed to one of the regret–minimization subroutines.
In the absence of interference, the SINR reduces to the signal-to-noise ratio (SNR). 
In what follows, we formally present the signal modeling and processing, interference modeling and filtering, and output SINR estimation.

\subsection{Signal Model and Match Filtering }

Let radar $i \in \Ical$ adopt a transmission strategy $\sigma^i$, consisting of a sequence of starting frequencies $(f^i_k)_{k=1}^{K^i}$ and a sequence of start times $(t^i_k)_{k=1}^{K^i}$ for each of the $K^i$ transmitted chirps. Consider $\mathcal{L} := \{1, \ldots, L\}$ targets within its field of view, with range vector $(r^i_\ell)_{\ell=1}^{L}$ and velocity vector $(\dot{r}^i_\ell)_{\ell=1}^{L}$. 
Approaching targets satisfy $\dot{r}^i_\ell < 0$.

The transmitted LFM signal at the $k$-th chirp is modeled as:
\begin{equation*}
    s^i_{\text{tx}}[t,k] = \begin{cases} e^{j 2\pi \left( f^i_k (t - t^i_k) + \frac{1}{2} \alpha^i (t - t^i_k)^2 \right)}, \quad  & t^i_k \leq t < t^i_k + T^i_a , \\ 0 ,& \text{otherwise}
\end{cases} 
\end{equation*}
and the received echo from the $\ell^{th}$ target is then given by:
\begin{equation*}
    s^i_{\text{rx}}[t,k] = \begin{cases} a^i_\ell s^i_{\text{tx}}[t-\tau^i_{k,\ell},k],& t^i_k + \tau^i_{k,\ell} \leq t < t^i_k + \tau^i_{k,\ell} \\
    &+ T^i_a, \\ 0, \quad \quad  & \text{otherwise},  
    \end{cases}
\end{equation*}
where $a^i_\ell$ is the complex target coefficient and $\tau^i_{k,\ell}$ is the round trip delay at chirp $k$, given by
$ \tau^i_{k,\ell} = \frac{2 }{c} \left[ r^i_\ell + \dot{r}^i_\ell    (k\Tpri^i - \Tpri^i + t^i_k - t^i_1) \right],$
with $c$ denoting the speed of light.

% Let the signal be sampled at Nyquist rate $f_s$, denoting the fast time sampling vector of dimension $N_t := \Tpri^i * f_s$ as $\boldsymbol{t}$, we have the signal,  

The echo signal, mixed with the transmitted signal, yields the system output:
\[
\begin{aligned}
    y^i[t,k]
    &=  s^i_{\text{tx}}[t,k]  \sum_{\ell=1}^{L}  s^{i*}_{\text{rx},\ell}[t,k] \\[3pt]
    &\approx \sum_{\ell=1}^{L} \tilde{a}^i_\ell   
    e^{  \jmath 2\pi   \left( f^i_k \tau^i_{k,\ell} + \alpha^i \tau^i_{k,\ell} t \right) }, 
\end{aligned}
\]
where constant terms are consolidated in the complex target coefficient $\tilde{a}^i_\ell$ and the coarse range migration across chirps~\cite{stettiner2023fmcw} is negligible due to small chirps with short bandwidth $B^i$.

To decouple coarse and fine range components, the target range is decomposed as
     $r^i_\ell + \dot{r}^i_\ell    (k\Tpri^i - \Tpri^i + t^i_k - t^i_1) 
     = \bar{r}^i_\ell + \left[ \epsilon^i_\ell + \dot{r}^i_\ell    (k\Tpri^i - \Tpri^i + t^i_k - t^i_1) \right],$
where $\bar{r}^i_\ell$ denotes the coarse range center of $\ell^{th}$ target with bin width $\frac{c}{2B^i}$, and $\epsilon^i_\ell \in \left[-\frac{c}{4B^i}, \frac{c}{4B^i}\right]$ is the fine offset.

Then the received echo becomes:
\[
\begin{aligned}
\quad \ y^i[t,k] &= \sum_{\ell=1}^{L}   \tilde{a}^i_\ell   
    e^{  \jmath 2\pi   \left( f^i_k \tau^i_{k,\ell} + \alpha^i \tau^i_{k,\ell} t \right) } \\ 
    & =  \sum_{\ell=1}^{L}  \tilde{a}^i_\ell    e^{j 2\pi \left[f_R t + f_D (k\Tpri^i - \Tpri^i + \Delta t_k^i) + \Delta f_k^i \tau^i_{k,\ell}   \right]} 
\end{aligned}
\]
where the hopping frequency offset is defined as $\Delta f_k^i = f_k^i - f_c$ and the time shift as $\Delta t_k^i = t_k^i - t_1^i$. The coarse range frequency is $f_R := \alpha^i\frac{2 r_\ell^i}{c}$, and the Doppler frequency is $f_D := f_c\frac{2 \dot{r}_\ell^i}{c}$. Without frequency hopping and time shifting (i.e., $\Delta f_k^i = 0$ and $\Delta t_k^i = 0$), the signal reduces to the standard form. Otherwise, the term $\Delta f_k^i \tau_{k,\ell}^i$ enables finer range resolution \cite{pan2025game}, while the term $\Delta t_k^i$ introduces a nonuniform pulse repetition frequency (PRF). Note that this nonuniform PRF will yield high
sidelobes in the Doppler spectrum and can be solved by advanced algorithms such as compressed sensing (CS) or joint range-Doppler processing \cite{xu2023automotive} which is another important topic and is out of the scope of this paper. Since Doppler resolution is primarily determined by the coherent processing interval (CPI), the time shift $\Delta t_k^i$, constrained within $\Tpri^i$, does not change the resolution but increases the likelihood of interference avoidance by randomizing chirp starting timings. Note that such nonuniform transmission along slow time significantly reduces the probability that a victim radar experiences persistent interference compared with uniform transmission. In practical deployments, when two radars share similar $\Tpri$ values, interference between their first overlapping chirps often persists across subsequent chirps under uniform transmission. 
% This phenomenon is exacerbated by the limited diversity of radar designs in the automotive industry, where a few major manufacturers adopt highly similar chirp configurations tailored for common applications such as corner radar for blind spot detection (BSD) and front radar for autonomous emergency braking (AEB).

%Note that the time shifting term allows the radar to avoid interference while introduces the random Doppler bias across chirps, which leads to ambiguity to the Doppler frequency estimation. 

Let $\mathbf{Y}^i \in \C^{N_t^i \times K^i}$ be the down-sampled and low-pass filtered ADC matrix for radar $i$ with columns indexed by chirp $k$ and rows by fast-time samples $n$. We process $\mathbf{Y}^i$ in three steps, i.e., fast-time processing, frequency hopping and time shifting compensation, nonuniform slow-time processing, yielding coarse range, fine range, and Doppler (velocity) estimates.
\paragraph*{1) Coarse-range processing}
Apply a fast-time coherent integration processing per chirp:
\begin{equation}
    \mathbf{Y}^i_{\mathrm{cr}} = \mathbf{R}^i   \mathbf{Y}^i,
\quad
    R^i_{n,m} = e^{-j \frac{2\pi}{N_t^i} n m},\ \ n,m=0,\ldots,N_t^i-1,
\label{eq:coarse_range_fft}
\end{equation}
which maps fast time to coarse range bins $r_m = \bar{r}^i_\ell$.

\paragraph*{2) Joint frequency-time compensation}
For each coarse bin $m$, chirp $k$, and a candidate fine-range offset $\epsilon_p \in [-\frac{c}{4B^i}, \frac{c}{4B^i }]$, construct
\begin{equation}
\begin{aligned}
    & \quad  Y_{cp}[m,k,p] \\
    & = Y^i_{cr}[m, k]  e^{j2\pi \left( \Delta f^i_k (\frac{2 r_m }{c} + \frac{2  \epsilon_p }{c} + \frac{2v_q (k\Tpri^i - \Tpri^i + \Delta t_k^i)}{c} ) \right) },     
\label{eq:hop_time_comp}
\end{aligned}
\end{equation}
where $\Delta f_k^i$ and $\Delta t_k^i$ are known by radar $i$'s receiver, $v_q$ is a velocity-grid value used when forming a joint $(\epsilon_p,v_q)$ hypothesis. The exponential term aligns inter-chirp phases introduced by frequency hopping and time-shifting so that fine-range information becomes coherent across $k$.

\paragraph*{3) Nonuniform slow-time processing}
For each $p$ and $q$-grid used in \eqref{eq:hop_time_comp}, apply the slow-time nonuniform matching along chirp index $k$:
\begin{equation}
    \mathbf{Y}^i_{\mathrm{out}}[:,q,p] 
    = \mathbf{Y}_{\mathrm{cp}}[:, :, p]   \mathbf{D}^i(:,q)
\label{eq:doppler_fft}
\end{equation}
and
\begin{equation}
     D^i_{k,q} = e^{-j \frac{2\pi(k\Tpri^i - \Tpri^i + \Delta t_k^i)}{K^i\Tpri^i} q},
\end{equation}
which forms a Range--Doppler (and fine-range) cube $Y^i_{\mathrm{out}}[m,q,p]$. We then select the global peak
$
    (m^*, q^*, p^*) 
    = \arg\max_{m,q,p} \left|Y^i_{\mathrm{out}}[m,q,p]\right|,
$
and read off the estimates$
    \hat r = r_{m^*} + \epsilon_{p^*}, \hat v \leftrightarrow q^* .$

\subsection{Interference Model and SINR Estimation}

At radar $i$'s receiver, the target echoes may be corrupted by interference from other interfering radars $o \in \text{neighbor}(\{i\})$, each transmitting at its own starting frequency $f^o_k$ and chirp slope $\alpha^o$.
\begin{equation*}
s^o_{tx}[t,k] = \begin{cases}  & e^{j 2\pi \left( f^o_k (t - \Delta_{io} - t^o_k) +  \frac{1}{2}  \alpha^o (t - \Delta_{io} - t^o_k)^2 \right)},  \\ &  \quad \quad \qquad  t^i_k \leq t < t^i_k + T^i_a , \\ & 0 , \quad \quad  \text{otherwise}.
\end{cases}
\end{equation*}

After dechirping (conjugate mixing with the victim's own TX reference), any interferer whose signal temporally overlaps with the victim's active window produces the beat signal
\begin{equation*}
    y^o[t,k] \approx \tilde{a}^o \, e^{j 2\pi \left[ (f^o_k - f^i_k)\, t \;+\; \frac{1}{2}(\alpha^o - \alpha^i)\, t^2 \right]},
\end{equation*}
where $\tilde{a}^o$ absorbs path loss and constant phase terms.
Interference corrupts the victim's measurement whenever this beat falls within the receiver's IF bandwidth, i.e., when the instantaneous beat frequency
$f_{\text{beat}}(t) = (f^o_k - f^i_k) + (\alpha^o - \alpha^i)\,t$
satisfies $|f_{\text{beat}}(t)| \leq f_{\text{IF,max}}$ for some $t$ during the active window.
This condition encompasses not only the ``chirp-crossing'' case where two ramps intersect ($f^o_k \neq f^i_k$ but slopes cause a crossing), but also the case where an aggressor chirp resides \emph{entirely within} the victim's configured bandwidth, producing a low beat frequency throughout the chirp duration.
In our simulations, the dechirping operation is applied to the full received waveform and the subsequent low-pass filter retains all beat components within the IF passband, thereby faithfully capturing both interference mechanisms~\cite{li2024performance}.

The overall received signal at radar $i$'s is then:
\begin{equation}
    \hat{\mathbf{y}}^i = \mathbf{y}^i  + \sum_{ o \in \mathbf{N}i}\mathbf{y}^o + \mathbf{e}^i,
\end{equation}
where $\mathbf{y}^i =  \{ y^i[t,k] \}$ is the desired (analogue) signal, (we have used $\mathbf{Y}^i$ to denote the down-sampled and low-pass filtered digital signal) $\sum_{ o \in \mathbf{N}i} \mathbf{y}^o$ is the interference signal aggregation from all the neighbors of radar $i$ in $\mathcal{G}$, i.e., $ o \in \mathbf{N}i \Rightarrow (o \to i) \in \mathcal{E}$, and $\mathbf{e}^i$ is the white noise. Note that here we omit the CPI index.

Let $\mathrm{P}(\cdot)$ denote average signal power. The theoretical \gls{sinr} for radar $i$ at CPI $\tau$ is expressed as:
\begin{equation}
    \mathrm{SINR}^i_\tau = \dfrac{\mathrm{P}(\mathbf{y}^i_\tau)}{\mathrm{P}( \sum_{ o \in \mathbf{N}i} \mathbf{y}^o_\tau) + \mathrm{P}(\mathbf{e}^i_\tau)}
\end{equation}

% Suppose chirps are divided into episodes $\boldsymbol{\tau} = \{k_{\tau-1} + 1, \ldots, k_\tau\}$ with $\tau = 1, \ldots, \mathcal{T}$ and $k_0 = 0$. Then, the average SINR at radar $i$ over episode $\tau$ is:
% \[
% \mathrm{SINR}^i_\tau = \frac{1}{|\boldsymbol{\tau}|} \sum_{k \in \boldsymbol{\tau}} \frac{\mathrm{P}(\mathbf{y}^i_k)}{\mathrm{P}(\mathbf{y}^o_k) + \mathrm{P}(\mathbf{e}^i_k)}.
% \]

We assume that each radar estimates per-action SINR using its own \gls{rx} pipeline at the end of every CPI, without an extra side-channel required. 
The actual computation requires interference detection (e.g., thresholding~\cite{li2024performance}), and separating the signal into a clean part $\tilde{\mathbf{y}}^i_\tau$ and an interference part $\tilde{\mathbf{y}}^o_\tau$. 
In practice, due to limited observability, the ground truth \gls{sinr} values are infeasible. We defer the estimation method to Section \ref{sec:solconceptalgo}.

% where the average noise power $N_0 = k\mathrm{TF}$, $k$ is the Boltzman constant, $\mathrm{T}$ is the antenna temperature, and $\mathrm{F}$ is the environment temperature.  
% \begin{equation}
%      =\frac{\sigma P_t G^2 \lambda^2 T_c}{(4 \pi)^3 R^4 k \mathrm{T F} }
% \end{equation}
% The received signal power $P_r$ and interference power $P_{\texttt{o}}$ can be expressed as
% \begin{equation}
%     P_r=\frac{P_t G^2 \lambda^2 a}{(4 \pi)^3 R^4}, \quad P_{\texttt{o}}=\frac{P_t G^2 \lambda^2}{(4 \pi)^2 R^2}, 
% \end{equation}
% where $P_t$ is the transmit power, $\lambda=c / f_0$ is the wavelength, $G$ is the antenna gain.

\section{Solution Concepts and Algorithm Design}\label{sec:solconceptalgo}

\subsection{Frequency-time Scheduling with Mixed Strategy}

We equip each radar $i$ with an adaptable mixed strategy $p^i_{\tau}$ throughout CPI $\tau=1,\ldots,\mathcal{T}$.
A mixed strategy is a probability distribution over the (deterministic) strategy class $\Sigma^i$, i.e.,
$p^i_{\tau}\in \Delta(\Sigma^i):=\{p:\Sigma^i\to[0,1]\mid \sum_{\sigma\in\Sigma^i} p(\sigma)=1\}$.
Before the start of CPI $\tau$, radar $i$ samples a transmission strategy $\sigma^i_\tau\sim p^i_\tau(\cdot)$ and then executes it deterministically across the $K^i$ chirps, which induces a per-chirp sequence of time–frequency pairs $\{(f_k^i,t_k^i)\}_{k=1}^{K^i}$.
In essence, mixed strategies generate diverse chirp sequences that sustain target detection while mitigating mutual interference.

Let $U_i(\sigma^i_\tau,\sigma^{-i}_\tau)$ be radar $i$’s CPI utility \eqref{eq:utility}, computed from the RX pipeline as $u( 10\log_{10}(\mathrm{SINR}^i_\tau) )$, where the superscript $-i$ denotes the strategy profile of all players except player~$i$, i.e., $\sigma^{-i}_\tau := (\sigma^j_\tau)_{j \in \Ical \setminus \{i\}}$.
The \emph{expected} CPI utility under independent mixed strategies $p^i_\tau$ and $p^{-i}_\tau$ is
\begin{equation}
\label{eq:exp-utility}
\bar U_i(p^i_\tau,p^{-i}_\tau)
=
\mathbb{E}_{\sigma^i_\tau\sim p^i_\tau,\ \sigma^{-i}_\tau\sim p^{-i}_\tau}
\big[U_i(\sigma^i_\tau,\sigma^{-i}_\tau)\big].
\end{equation}
Each radar updates $p^i_\tau$ online, from local measurements only, to maximize its long-term utility. In classical noncooperative game theory, when players act independently according to their own mixed strategies, the system may reveal a \emph{Nash equilibrium} (NE)\cite{bacsar1998dynamic}.
\begin{definition}\label{def:ne} 
    For the $\mathcal{T}$-stage repeated game $\Gamma = \{ \Ical, \Sigma, \mathcal{G}, \{U_i\}_{i \in \Ical}\}$, a joint strategy profile $(p^{1\*},\ldots,p^{I\*})$ constitutes a Nash equilibrium if no radar can unilaterally improve its expected utility, i.e.,
\begin{equation}
\label{eq:NE}
    \bar U_i(p^{i *},p^{-i *})
    \ge
    \bar U_i(p^i,p^{-i *}), 
    \quad \forall p^i\in\Delta(\Sigma^i),\ \forall i\in\Ical.
\end{equation} 
\end{definition}
Intuitively, under an NE each radar’s mixed strategy is an optimal response to the others’ strategies, and no player has incentive to deviate.
However, in radar interference scenarios, purely independent mixed strategies are often insufficient to ensure interference-free operation due to a variety of reasons discussed in \cite{pan2025game}. 
Even when a Nash equilibrium exists, it often suffers from several fundamental limitations:
(i) an equilibrium selection problem, since multiple equilibria may coexist and independent learners can converge to inconsistent or colliding strategies;
(ii) suboptimality, as individual optimality does not imply maximal collective performance, often resulting in degraded overall SINR and inefficient spectrum utilization; and
(iii) computational intractability, as finding a Nash equilibrium is known to be a PPAD-complete problem~\cite{thecomplexityofne}.

\subsection{Correlated Equilibria and Coarse Correlated Equilibria}

To overcome these limitations, a richer notion of equilibrium allows players’ actions to be statistically \emph{correlated} through shared experience or endogenous signals.
In this setting, the joint distribution over strategies, denoted by $\pi \in\Delta(\prod_{i \in \Ical}\Sigma^i)$, need not decompose into a product of marginals. \emph{Correlated strategies} capture coordination patterns that emerge through repeated interaction and learning, enabling the system to approach equilibria with higher collective utility. Rather than assuming independent sampling of strategies, we let each radar draw its strategy recommendation from $\pi$ and can decide whether to follow it or deviate.
This framework generalizes the aforementioned setting, which corresponds to the special case in which $\pi$ factorizes as the product of independent marginals.

\begin{definition}
For the $\mathcal{T}$-stage repeated game $\Gamma = \{ \Ical, \{ \Sigma^i \}_{i \in \Ical }, \mathcal{G}, \{U_i\}_{i \in \Ical}\}$, 
let $\pi \in \Delta(\prod_{i\in\Ical}\Sigma^i)$ denote a joint probability distribution over the players’ strategy profiles.
Then:
\begin{itemize}
    \item[a)] $\pi$ is a {Correlated Equilibrium (CE)} if, for every radar $i \in \Ical$ and for all $\sigma^i, \sigma^{i\prime} \in \Sigma^i$, 
    \begin{equation}\label{eq:ce}
        \mathbb{E}_{\sigma^{-i} \sim \pi(\cdot \mid \sigma^i)}
        \left[ U_i(\sigma^i,\sigma^{-i}) \right]
        \ge
        \mathbb{E}_{\sigma^{-i} \sim \pi(\cdot \mid \sigma^i)}
        \left[ U_i(\sigma^{i\prime},\sigma^{-i}) \right].
    \end{equation}
    That is, once a joint strategy $\sigma$ is drawn from $\pi$ and radar $i$ learns its own recommendation $\sigma^i$, it cannot improve its expected utility by deviating to another strategy $\sigma^{i\prime}$.

    \item[b)] $\pi$ is a {Coarse Correlated Equilibrium (CCE)} if, for every radar $i \in \Ical$ and for all $\sigma^{i\prime} \in \Sigma^i$, 
    \begin{equation}\label{eq:cce}
        \mathbb{E}_{\sigma \sim \pi}
        \left[ U_i(\sigma^i,\sigma^{-i}) \right]
        \ge
        \mathbb{E}_{\sigma \sim \pi}
        \left[ U_i(\sigma^{i\prime},\sigma^{-i}) \right].
    \end{equation}
    In this case, radar $i$ decides whether to deviate \emph{before} seeing its recommended strategy, and such deviation cannot increase its expected utility.
\end{itemize}
\end{definition}

A CCE is a relaxation of CE in which deviation decisions are made \emph{before} observing any recommendation. 
Hence, every CE is a CCE, but the converse does not always hold; the set of CCE is larger and typically easier to approach in practice.  

In general, computing CE or CCE directly requires solving systems of linear inequalities of the form~\eqref{eq:ce}–\eqref{eq:cce}, whose size grows exponentially with the number of radars and the complexity of their strategy spaces~\cite{roughgarden2010algorithmic}. 
However, centralized planning is infeasible in the radar-to-radar interference setting as it requires information exchange infrastructure. 
Therefore, we resort to \emph{no-regret learning} to obtain these equilibria, where each radar independently updates its mixed strategy based on local utility feedback.

\begin{definition}
   The external regret for radar $i$ after $\mathcal{T}$ CPIs is defined, given a sequence of played strategy profiles $\{\sigma^i_{\tau}, \sigma^{-i}_{\tau}\}_{\tau = 1}^{\mathcal{T}}$, as:
\begin{equation}
\label{eq:external_regret}
\mathcal{R}_{\mathrm{ext}}^i(\mathcal{T})
=\max_{\sigma^{i,\star} \in \Sigma^i}
\sum_{\tau=1}^{\mathcal{T}}
\Big[
U_i(\sigma^{i,\star}, \sigma^{-i}_\tau)
-
U_i(\sigma^i_\tau, \sigma^{-i}_\tau)
\Big].
\end{equation}
\end{definition}
 External regret measures how much worse a radar performs compared to the best constant strategy in hindsight.
A learning process achieves \emph{no external regret} if $\mathcal{R}_{\mathrm{ext}}^i(\mathcal{T})/\mathcal{T} \to 0$ as $\mathcal{T}\to\infty$, meaning radar $i$ performs asymptotically at least as well as the best constant strategy in hindsight.

The notion of \emph{internal regret}, also known as \emph{swap regret}, refines and generalizes this benchmark by considering deviations from each played strategy to alternative ones according to a deterministic \emph{swap mapping} (also known as a \emph{modification rule}).

\begin{definition}
    Let $\phi^i: \Sigma^i \to \Sigma^i$ be a mapping that prescribes, for every strategy $\sigma^i \in \Sigma^i$, an alternative strategy $\phi^i(\sigma^i)$ to which radar $i$ would have switched whenever $\sigma^i$ was selected.
The cumulative internal (swap) regret of radar $i$ after $\mathcal{T}$ CPIs is defined, given a sequence of fixed strategy profiles $\{\sigma^i_{\tau}, \sigma^{-i}_{\tau}\}_{\tau = 1}^{\mathcal{T}}$, as
\begin{equation}
\label{eq:swap_regret}
\mathcal{R}_{\mathrm{int}}^i(\mathcal{T})
=\max_{\phi^i:\Sigma^i\to\Sigma^i}
\sum_{\tau=1}^{\mathcal{T}}
\Big[
U_i(\phi^i(\sigma^i_\tau), \sigma^{-i}_\tau)
-
U_i(\sigma^i_\tau, \sigma^{-i}_\tau)
\Big].
\end{equation}
\end{definition}
A learning process is said to have \emph{no internal regret} if $\mathcal{R}_{\mathrm{int}}^i(\mathcal{T})/\mathcal{T} \to 0$ as $\mathcal{T} \to \infty$.
Intuitively, this means that no systematic replacement rule (swapping each previously played strategy $\sigma^i$ with another $\phi^i(\sigma^i)$) would have yielded higher average utility. 

\subsection{No-Regret Learning }

A key result from learning in games~\cite{roughgarden2010algorithmic,hart2000simple,cesa2006prediction}
establishes a fundamental bridge between regret minimization and equilibrium concepts. 
When players iteratively update their mixed strategies to minimize regret, 
the empirical distribution of play converges to an equilibrium notion consistent with the form of regret being minimized.
Specifically, minimizing \emph{external regret} leads to convergence to a \emph{coarse correlated equilibrium} (CCE),
while minimizing the stronger \emph{internal (swap) regret} yields convergence to a \emph{correlated equilibrium} (CE).

Let the empirical joint distribution of strategies after $\mathcal{T}$ \gls{cpi} be
\begin{equation}\label{eq:emp}
    \bar{\pi}(\sigma) 
    := \frac{1}{\mathcal{T}}\sum_{\tau=1}^{\mathcal{T}}\mathds{1}\{\sigma_\tau=\sigma\},
    \qquad 
    \sigma=(\sigma^1,\ldots,\sigma^I)\in\prod_{i\in\Ical}\Sigma^i,
\end{equation}
where $\sigma_\tau$ denotes the joint strategy profile played at CPI $\tau$.

\begin{lemma}\label{lemma:cenoregret}
Suppose each radar $i\in\Ical$ achieves average internal (swap) regret at most $\varepsilon$, i.e.,
\begin{equation}\label{eq:swap_regret_bound}
\frac{1}{\mathcal{T}}
\max_{\phi^i:\Sigma^i\to\Sigma^i}
\sum_{\tau=1}^{\mathcal{T}}
\Big[
U_i(\phi^i(\sigma^i_\tau),\sigma^{-i}_\tau)
-U_i(\sigma^i_\tau,\sigma^{-i}_\tau)
\Big]
\le \varepsilon,
\end{equation}
where $\phi^i$ is a \emph{swap mapping} prescribing, for each strategy $\sigma^i$, an alternative $\phi^i(\sigma^i)$. 
Then, the empirical distribution $\bar{\pi}$ is an $\varepsilon$-correlated equilibrium.
\end{lemma}

\begin{proof}
Fix $i\in\Ical$ and two strategies $\hat{\sigma}^i,\tilde{\sigma}^i\in\Sigma^i$.  
Define a pairwise swap mapping 
$\phi^i_{\hat{\sigma}^i\to\tilde{\sigma}^i}(\sigma^i)
= \tilde{\sigma}^i$ if $\sigma^i=\hat{\sigma}^i$ and $\sigma^i$ otherwise.
Substituting $\phi^i_{\hat{\sigma}^i\to\tilde{\sigma}^i}$ into~\eqref{eq:swap_regret_bound} yields
\begin{equation*}
\frac{1}{\mathcal{T}}\sum_{\tau=1}^{\mathcal{T}}
\mathds{1}\{\sigma^i_\tau=\hat{\sigma}^i\}
\Big[U_i(\hat{\sigma}^i,\sigma^{-i}_\tau)-U_i(\tilde{\sigma}^i,\sigma^{-i}_\tau)\Big]
\ge -\varepsilon.
\end{equation*}
Recognizing the left-hand side as an expectation under $\bar{\pi}$ gives~\eqref{eq:CE_weighted}.  
\begin{equation}\label{eq:CE_weighted}
\mathbb{E}_{\sigma\sim\bar{\pi}}
\Big[
\mathds{1}\{\sigma^i = \hat{\sigma}^i\}\big(U_i(\sigma^i,\sigma^{-i})-U_i(\tilde{\sigma}^i,\sigma^{-i})\big)
\Big]\ge -\varepsilon,
\end{equation}
Dividing by $\bar{\pi}(\sigma^i)$ (if nonzero) yields the conditional form~\eqref{eq:CE_conditional}. 
\begin{equation}\label{eq:CE_conditional}
\mathbb{E}_{\sigma\sim\bar{\pi}(\cdot\mid\sigma^i)}\big[U_i(\sigma^i,\sigma^{-i})\big]
\ge
\mathbb{E}_{\sigma \sim\bar{\pi}(\cdot\mid\sigma^i)}\big[U_i(\sigma^{i\prime},\sigma^{-i})\big]
-\frac{\varepsilon}{\bar{\pi}(\sigma^i)}.
\end{equation}
Since this holds for all $i$, $\sigma^i$, and $\sigma^{i\prime}$, $\bar{\pi}$ satisfies the $\varepsilon$-CE inequalities.
\end{proof}

\begin{corollary}\label{cor:cce}
External regret is a special case of internal (swap) regret, obtained by restricting the swap mapping to a constant function,
\begin{equation*}
\phi^i(\sigma^i) \equiv \sigma^{i\prime}, \quad \forall \sigma^i \in \Sigma^i.
\end{equation*}
Hence, if each radar achieves vanishing external regret, the empirical distribution $\bar{\pi}$ satisfies the relaxed inequalities of an $\varepsilon$-\emph{coarse correlated equilibrium} (CCE):
\begin{equation*}
\mathbb{E}_{\sigma\sim\bar{\pi}}[U_i(\sigma^i,\sigma^{-i})]
\ge
\mathbb{E}_{\sigma\sim\bar{\pi}}[U_i(\sigma^{i\prime},\sigma^{-i})]
-\varepsilon, \quad \forall i\in\Ical.
\end{equation*}
Consequently, minimizing swap regret implies minimizing external regret and guarantees convergence to both CE and CCE.
\end{corollary}

In the context of radar interference avoidance, these results imply that when each radar locally minimizes its own regret, the empirical average of collective radar behavior asymptotically converges to a CE—a decentralized operating regime that balances interference avoidance with detection performance without side-channel coordination.

\subsection{Algorithmic Design under Semi-Bandit Feedback}

\begin{algorithm*}[ht]
\caption{\texttt{No-Regret Hopping} for radar $i\in\Ical$}\label{algo:noregrettransmit}
\begin{algorithmic}[1]
\Require initial strategy $p^{i}_{1}=\mathrm{Unif}(\Sigma^{i})$; utility buffer $\mathcal{U}\leftarrow\varnothing$; block length $\ell$; number of chirps $K^i$, parameters $\{ \eta_\tau, \gamma_\tau \}_{\tau=1}^{\mathcal{T}}$ 

\For{$\tau=1:\mathcal{T}$}
    \State {Transmission scheduling:} 
     \hspace{\algorithmicindent}
        $(f^{i}_{k},t^{i}_{k})_{k\in\boldsymbol{\tau}}
        \leftarrow \texttt{Stochastic Round-Robin}\big(p^{i}_{\tau}, \ell, K^i\big)$
        \Comment{(\emph{ Subroutine~\ref{algo:sub1}})}
    \State Sweep subband $[f^{i}_{k},f^{i}_{k}+B]$ over $[t^{i}_{k},  t^{i}_{k}+T^{i}_{a}]$ for each chirp $k\in\boldsymbol{\tau}$
    \State {Local feedback:}
        estimate $\overline{\mathrm{SINR}}^{i}_{\tau}(f,t)$ (and $\overline{\mathrm{SNR}}^{i}_{\tau}(f,t)$ when available) for all $(f,t)$ used
    \State Form CPI utility $\bar U^{i}_{\tau}$ from \gls{rx} pipeline (cf. \eqref{eq:utility}), $\mathcal{U} = \mathcal{U} \bigcup \{  \bar U^{i}_{\tau} \}$
    \State {Strategy update:} \hspace{\algorithmicindent}
        $\left(p^{i}_{\tau+1}, \Theta^i_{\tau+1} \right) \leftarrow\texttt{Regret-Minimization}\big(\mathcal{U},  p^{i}_{\tau}, \Theta^i_{\tau}, \eta_\tau, \gamma_\tau \big)$
         \Comment{(\emph{Subroutine~\ref{algo:sub2_ext} or \ref{algo:sub2_int}})}
\EndFor
\end{algorithmic}
\end{algorithm*}

In this subsection, we develop the decentralized learning procedure from the perspective of an individual radar. 
Each radar observes only its own noisy measurement outcomes and lacks access to the global joint utility. 
Consequently, the learning process follows a form of \emph{semi-bandit feedback}, 
where the local utility of a strategy is inferred from locally observable signal statistics obtained at the end of every CPI. It is termed {semi-bandit} because the observation from a played strategy can indirectly reveal information about the utilities of other, unplayed strategies.

We propose Algorithm~\ref{algo:noregrettransmit} as the CPI-level driver that each radar executes online, which unfolds into three main procedures.
\paragraph{Transmission Scheduling}
At the beginning of CPI~$\tau$, radar $i$ holds a mixed strategy $p^i_\tau\in\Delta(\Sigma^i)$ over strategy space.
Subroutine~\ref{algo:sub1} (\emph{Stochastic Round–Robin}) then instantiates the concrete per–chirp sequence $\{(f^i_k,t^i_k)\}_{k\in\boldsymbol{\tau}}$ by sampling a starting pair once per \emph{block of chirps} and then cycling through indices within that block. 
This meta routine preserves temporal diversity within a CPI, and gives freedom for the design of the transmission subroutine, which seeks to diversify chirps but also avoid collisions.

\floatname{algorithm}{Subroutine}
\begin{algorithm}[ht]
\caption{\texttt{Stochastic Round-Robin}}\label{algo:sub1}
\begin{algorithmic}[1]
\Require from $p^i_\tau$ get sampler $q^i_\tau$ over  $(f,t)\in\mathcal{A}$ for CPI $\tau$ at radar $i$; block length $\ell$; number of chirps $K^i$
% \Ensure sequence $(f^i_k,t^i_k)_{k\in\boldsymbol{\tau}}$
\For{$b=1$ to $\lceil K^i / \ell \rceil$}
    \State Sample starting pair $(f_s,t_s)\sim q^i_\tau(\cdot)$
    \For{$m=0$ to $\ell-1$}
        \State $k \leftarrow (b-1)\ell + m + 1$
        \State $f^i_k \leftarrow f_{(s+m)\bmod A_f}$,\quad
               $t^i_k \leftarrow t_{(s+m)\bmod A_t}$
    \EndFor
\EndFor
\State \textbf{return} $(f^i_k,t^i_k)_{k = 1}^{K^i}$
\end{algorithmic}
\end{algorithm}

\paragraph{Local Feedback and Signal Processing}
At the end of each CPI, the receiver computes statistics $\overline{\mathrm{SINR}}^i_\tau(f,t)$ (and $\overline{\mathrm{SNR}}^i_\tau$ when no interference is detected), and aggregates them into a bounded CPI utility $\bar U^i_\tau\in[0,1]$, appended to the history $\mathcal{U}$.
Specifically, we assume that at the end of CPI~$\tau$, radar~$i$ computes its signal statistics from the down-sampled and low-pass filtered ADC matrix $\mathbf{Y}^i_\tau$, whose columns $\mathbf{Y}^i_{\tau,k}$ correspond to the received baseband signal of chirp~$k$. 
This correspondence allows the radar to associate each chirp with its transmitted time–frequency configuration $(f^i_k,t^i_k)$ and thus estimate, for every pair $(f,t)$, the empirical \gls{sinr} as
\begin{equation} \label{eq:sinr}
\overline{\mathrm{SINR}}^i_\tau(f, t)
=
\dfrac{ \sum_{k }
\frac{\mathrm{P}(\tilde{\mathbf{Y}}^i_{\tau, k} ) \mathds{1}\{f^i_k=f,  t^i_k=t\}}
{\mathrm{P}( \sum_{o \in \mathbf{N}(i)}\tilde{\mathbf{Y}}^o_{\tau, k} )+\mathrm{P}(\mathbf{E}^i_k)}}
{\sum_{k}\mathds{1}\{f^i_k=f,  t^i_k=t\}},
\end{equation}
where $\mathbf{E}^i_{\tau,k}$ denotes the receiver noise component and $\tilde{\mathbf{Y}}^o_{\tau,k}$ represents the interference signals identified from neighboring radars $o\in\mathbf{N}(i)$. 
We further assume that  
(i)~the average noise power can be estimated as $N_0=k_B T_0 F B_{\text{IF}}$, where $k_B$ is the Boltzmann constant, $T_0 = 290$~K is the standard reference temperature, $F$ is the receiver noise figure, and $B_{\text{IF}}$ is the IF bandwidth; and  
(ii)~interference components are detectable and separable by standard interference-identification techniques such as time-domain thresholding~\cite{bechter2017automotive_intf}, STFT-based analysis~\cite{uysal2020phase_coded_fmcw}, or CFAR detection~\cite{jin2019adaptive_noise_canceller}, enabling the estimation of $\mathrm{P}(\sum_{o\in\mathbf{N}(i)}\tilde{\mathbf{Y}}^o_{\tau,k})$. In our simulation, we implement the STFT-based thresholding and cell-averaging CFAR pipeline of~\cite{li2024performance} for per-chirp SINR estimation, with parameters: Hamming window length 64, overlap 48 samples, 64-point FFT, interference threshold factor $10^3$ (median + MAD), CA-CFAR with 2 guard cells, 12 training cells per side, and $P_{fa} = 10^{-3}$. We note that imperfect interference detection would degrade the utility estimate quality, slowing but not preventing convergence of the no-regret algorithm.

When interference-free chirps are detected, the corresponding signal-to-noise ratio (SNR) can be estimated as
\begin{equation}\label{eq:snr}
\overline{\mathrm{SNR}}^i_\tau(f, t)
=
\dfrac{\sum_{k}
\frac{\mathrm{P}(\mathbf{Y}^i_{\tau, k} )\mathds{1}\{f^i_k=f,  t^i_k=t,  \tilde{\mathbf{Y}}^i_{\tau, k}=\mathbf{Y}^i_{\tau, k}\}}
{\mathrm{P}(\mathbf{E}^i_{\tau, k})}}
{\sum_{k}\mathds{1}\{f^i_k=f,  t^i_k=t,  \tilde{\mathbf{Y}}^i_{\tau, k}=\mathbf{Y}^i_{\tau, k}\}}.
\end{equation}
These empirical SINR and SNR estimates constitute the observable performance feedback for radar~$i$ and serve as the input for its utility estimation and subsequent mixed-strategy updates.

\paragraph{Mixed-strategy update and regret minimization}
Subroutine~\ref{algo:sub2_ext} implements entropic Online Mirror Descent (OMD) \cite{pan2024variationalinterpretationmirrorplay,pan-tao23delay,shutian23erm} with a \emph{bandit} (one–point) utility estimator on the strategy space $\Sigma^i$.
Since each CPI samples exactly one pure strategy $\sigma^i_\tau\sim p^i_\tau$, we form the importance–weighted estimate
$\widehat U^i_\tau(\sigma)=\bar U^i_\tau/p^i_\tau(\sigma^i_\tau)$ if $\sigma=\sigma^i_\tau$ and $0$ otherwise, 
which is unbiased for the linearized objective on $\Delta(\Sigma^i)$.
The dual states $\Theta^i$ are updated with stepsize $\eta_\tau$ and projected back to the simplex by a softmax; a small exploration mass $\gamma_\tau$ keeps all actions played with nonzero probability (controlling estimator variance and ensuring well–posedness).
% With $\eta_\tau\propto 1/\sqrt{\tau}$ (and any fixed $\gamma_\tau\in(0,1)$), the external regret against the best fixed pure strategy in $\Sigma^i$ vanishes at the canonical $O(\sqrt{\mathcal T\log|\Sigma^i|})$ rate.
% By the standard regret$\Rightarrow$equilibrium result, when all players run Subroutine~\ref{algo:sub2_ext} the empirical joint distribution of realized CPI sequences converges to the set of (approximate) CCE of the \emph{restricted} finite game whose admissible deviations are the strategies reachable through the sampler $\Pi$.

\begin{algorithm}[htbp]
\caption{\texttt{Regret-Minimization (External)}}\label{algo:sub2_ext}
\begin{algorithmic}[1]
\Require current strategy $p^i_\tau\in\Delta(\Sigma^i)$; CPI utility history $\mathcal{U}$; stepsize $\eta_\tau>0$; exploration $\gamma_\tau\in(0,1)$; \emph{dual state} $\Theta^i_\tau$ storing scores $z^i_\tau(\sigma)$ for all $\sigma\in\Sigma^i$

\State Estimation: for every $\sigma \in \Sigma^i$
\begin{align*}
\widehat{U}^i_\tau(\sigma)\leftarrow
\begin{cases}
\displaystyle \frac{\bar{U}^i_\tau}{p^i_\tau(\sigma)} \mathds{1} \{ \sigma=\sigma^i_\tau \}, &  \gamma_\tau \neq 0 \\
1 - \frac{1- \bar{U}^i_\tau}{p^i_\tau(\sigma)} \mathds{1} \{ \sigma=\sigma^i_\tau \},& \gamma_\tau = 0 .
\end{cases}
\end{align*}

\State Update: $$z^i_{\tau+1}(\sigma)\leftarrow z^i_\tau(\sigma)+\eta_\tau  \widehat{U}^i_\tau(\sigma) \quad \forall\sigma\in\Sigma^i. $$
 {Exploitation + exploration:}
\begin{align*}
    \tilde p^i_{\tau+1}(\sigma) & =\frac{\exp\{z^i_{\tau+1}(\sigma)\}}{\sum_{\sigma'}\exp\{z^i_{\tau+1}(\sigma')\}},\\
p^i_{\tau+1}(\sigma) & =(1-\gamma_\tau)  \tilde p^i_{\tau+1}(\sigma)+\frac{\gamma_\tau}{|\Sigma^i|}.
\end{align*}

\State \textbf{Return} $p^i_{\tau+1}$ and $\Theta^i_{\tau+1}\equiv\{z^i_{\tau+1}(\sigma)\}_{\sigma\in\Sigma^i}$
\end{algorithmic}
\end{algorithm}

Subroutine~\ref{algo:sub2_int} upgrades the guarantee from CCE to CE by minimizing \emph{swap} regret on $\Delta(\Sigma^i)$.
It maintains, for every source pure strategy $\sigma\in\Sigma^i$, a row-mixed-strategy $q^i_{\sigma}(\cdot)\in\Delta(\Sigma^i)$ that encodes how probability mass on $\sigma$ would be \emph{relabelled} to any alternative $\sigma'$.
A single CPI reward $\bar U^i_\tau$ yields the same importance–weighted estimator $\widehat U^i_\tau(\cdot)$ as above; then each row is updated by row–wise entropic OMD with a $p^i_\tau(\sigma)$ weighting, which is the canonical construction for no–swap–regret.
Collecting the rows forms a stochastic matrix $Q^i_{\tau+1}$; the next mixed strategy $p^i_{\tau+1}$ is its left stationary distribution, i.e., the solution to $p^i_{\tau+1}=p^i_{\tau+1}Q^i_{\tau+1}$, $\sum_\sigma p^i_{\tau+1}(\sigma)=1$.
% With standard step–size schedules and exploration, each player’s average swap regret vanishes, and the empirical joint play approaches the CE set of the same restricted game (reachable through $\Pi$).
% This delivers stronger stability (no profitable pairwise \emph{or} composite relabellings) without altering the CPI–level interface.

\begin{algorithm}[ht]
\caption{\texttt{Regret-Minimization (Internal)}}\label{algo:sub2_int}
\begin{algorithmic}[1]
\Require current strategy $p^i_\tau\in\Delta(\Sigma^i)$; CPI utility history $\mathcal{U}$; stepsize $\eta_\tau$; exploration $\gamma_\tau$; \emph{dual state} $\Theta^i_\tau$ storing row scores $z^i_{\tau,\sigma}(\cdot)$ and row-strategies $q^i_{\tau,\sigma}(\cdot)$ for each $\sigma\in\Sigma^i$.

\State {Estimation:} for every $\sigma \in \Sigma^i$
\begin{align*}
\widehat{U}^i_\tau(\sigma)\leftarrow
\begin{cases}
\displaystyle \frac{\bar{U}^i_\tau}{p^i_\tau(\sigma)} \mathds{1} \{ \sigma=\sigma^i_\tau \}, &  \gamma_\tau \neq 0 \\
1 - \frac{1- \bar{U}^i_\tau}{p^i_\tau(\sigma)} \mathds{1} \{ \sigma=\sigma^i_\tau \},& \gamma_\tau = 0 .
\end{cases}
\end{align*}

\State {Row-wise update:} For each source $\sigma\in\Sigma^i$ and destination $\sigma'\in\Sigma^i$,
\begin{align*}
z^i_{\tau+1,\sigma}(\sigma')\leftarrow
z^i_{\tau,\sigma}(\sigma')
+\eta_\tau  p^i_\tau(\sigma)
\widehat{U}^i_\tau(\sigma^{\prime}).
\end{align*}

\State Exploration + Exploitation:
\begin{align*}
\tilde q^i_{\tau+1,\sigma}(\sigma')
&=\frac{\exp\left\{ z^i_{\tau+1,\sigma}(\sigma') \right\}}
{\sum_{\rho\in\Sigma^i}\exp\left\{  z^i_{\tau+1,\sigma}(\rho)  \right\}},\\
q^i_{\tau+1,\sigma}(\sigma')
&=(1-\gamma_\tau)  \tilde q^i_{\tau+1,\sigma}(\sigma')
+\frac{\gamma_\tau}{|\Sigma^i|}.
\end{align*}

\State {Assemble transition matrix:} $Q^i_{\tau+1}[\sigma,\cdot]=q^i_{\tau+1,\sigma}(\cdot),$
\begin{align*}
p^i_{\tau+1}=p^i_{\tau+1}  Q^i_{\tau+1},
\quad
\sum_{\sigma\in\Sigma^i}p^i_{\tau+1}(\sigma)=1.
\end{align*}

\State \textbf{Return}
$p^i_{\tau+1}$ and
$\Theta^i_{\tau+1}\equiv\{z^i_{\tau+1,\sigma},  q^i_{\tau+1,\sigma}\}_{\sigma\in\Sigma^i}$.
\end{algorithmic}
\end{algorithm}
\floatname{algorithm}{Algorithm}

% In summary, Algorithm~\ref{algo:noregrettransmit} alternates, once per CPI, between \emph{generation} of a time–frequency sequence via Subroutine~\ref{algo:sub1} and a \emph{single} regret–based strategy update via Subroutine~\ref{algo:sub2_ext} (CCE) or Subroutine~\ref{algo:sub2_int} (CE), or any other plug-in ready regret-minimization schemes.
% All learning happens on the pure strategy simplex $\Delta(\Sigma^i)$, which keeps computation light and allows clean equilibrium statements. 
% The stochastic round-robin scheme provides the bridge from a low–dimensional strategy (only sample frequency-time pairs) to physically realizable chirp schedules; consequently, the convergence claims hold for the finite game restricted to the subset of mixed sequence strategies implementable by this scheme, faithfully reflecting the \gls{tx}/\gls{rx} architecture and feedback available at each radar.

In summary, Algorithm~\ref{algo:noregrettransmit} proceeds once per CPI by (i) \emph{generating} a time–frequency sequence via Subroutine~\ref{algo:sub1} and then (ii) performing a \emph{single} regret-based mixed-strategy update via Subroutine~\ref{algo:sub2_ext} (CCE), Subroutine~\ref{algo:sub2_int} (CE), or any other plug-in no-regret routine.
All learning occurs on the mixed-strategy simplex $\Delta(\Sigma^i)$, keeping computation light and enabling clean equilibrium statements.
The stochastic round-robin sampler bridges the low-dimensional mixed-strategy (sampling starting frequency-time pairs) to physically realizable chirp schedules; hence our convergence guarantees apply to the finite game restricted to the mixed sequence strategies implementable by this sampler, faithfully reflecting the \gls{tx}/\gls{rx} architecture and the per-CPI feedback available at each radar.

\subsubsection{Computational cost and real-time feasibility}
The per-CPI strategy update is lightweight because the action space is small ($A_f \times A_t = 21$). In our implementation, the external update (Subroutine~\ref{algo:sub2_ext}) accumulates action scores over the observed CPI history, while the internal update (Subroutine~\ref{algo:sub2_int}) forms a small group-level regret matrix; both operate on low-dimensional vectors and matrices and are negligible compared with receiver-side range-Doppler FFT, STFT-based interference detection, and CFAR processing that every radar already performs. The stochastic round-robin sampling (Subroutine~\ref{algo:sub1}) draws $\lceil K/\ell \rceil$ random samples per CPI. No inter-radar communication is required.

The update executes once between consecutive CPIs. With $K = 256$ and $T_{\text{PRI}} = 29.99\,\mu\text{s}$, one CPI lasts approximately 7.68\,ms, so convergence in ${\sim}10$ CPIs corresponds to ${\sim}77$\,ms of real time. At a closing speed of 260\,km/h, this amounts to only ${\sim}5.5$\,m of relative motion, well within the 200\,m interference range. Thus, in the tested scenarios, geometry changes are modest over the learning window, although the interference graph is recomputed every CPI to reflect the current topology.

\subsection{Regret Analysis and Convergence Guarantee}

\subsubsection{External/internal-regret bound}
We establish the sublinear external-regret bound in the number of CPI rounds $\mathcal{T}$ for 
Subroutine~\ref{algo:sub2_ext}, instantiated for any radar player $i\in\mathcal{I}$. 
In addition to the classical $\mathcal{O}(|\Sigma^i|^{1/2}\mathcal{T}^{1/2})$ bound that is widely seen in OMD style methods, 
our result demonstrates that casting an explicit constant exploration parameter will result in a larger dependence on time ($\mathcal{O}(\mathcal{T}^{2/3})$) but a milder dependence on the size of the strategy space ($\mathcal{O}(|\Sigma^i|^{1/3})$).
This tradeoff is particularly well-suited to our setting where the strategy space is large but the CPI horizon $\mathcal{T}$ is relatively limited.
The bound immediately yields vanishing
average regret $\mathcal{O} \big( \max_{i\in \Ical}\{ (|\Sigma^i| \log |\Sigma^i|)^{1/3} \mathcal{T}^{2/3}  \}\big)$ and the convergence of the empirical play to the CCE set (by Corollary~\ref{cor:cce}).

\begin{theorem}\label{thm:bandit_ext_regret}
Fix a radar $i \in \Ical$ running Algorithm~\ref{algo:noregrettransmit} with time-invariant parameters $\eta_\tau \equiv \eta$ and $\gamma_\tau \equiv \min \{ \gamma_{\max} , \gamma\}$, the expected cumulative external regret for Subroutine~\ref{algo:sub2_ext} against the best fixed strategy $\sigma^\star\in\Sigma^i$ obeys
\begin{equation}\label{eq:regret_bound_main}
\mathbb{E} \left[ \mathcal{R}^i_{\mathrm{ext}} (\mathcal{T})\right]
\le  \frac{\log |\Sigma^i|}{\eta} + \frac{\eta \mathcal{T}|\Sigma^i|}{ \min \{ \gamma_{\max}, \gamma \} }  + \min\{ \gamma_{\max}, \gamma\}\mathcal{T},
\end{equation}
where the constant $\gamma_{\max} \in (0,1)$.

\noindent\textbf{Case $\gamma = 0$ (no exploration):} The bound simplifies to
\begin{equation}\label{eq:regret_bound_main2}
\mathbb{E} \left[ \mathcal{R}^i_{\mathrm{ext}} (\mathcal{T})\right]
\le  \frac{\log |\Sigma^i|}{\eta} + \eta \mathcal{T}|\Sigma^i|  .
\end{equation}
In particular,
  $  \gamma\  :=\ \frac{ (|\Sigma^i|\log |\Sigma^i|)^{1/3}}{\mathcal{T}^{1/3}}, 
\eta\  :=\ \frac{(\log |\Sigma^i|)^{2/3}}{|\Sigma^i|^{1/3} \mathcal{T}^{2/3}}$ 
yields
\begin{equation*}
\begin{aligned}
    \mathbb{E}[\mathcal{R}^i_\mathcal{\mathrm{ext}} (\mathcal{T})]
  =
\mathcal{O} \Big(|\Sigma^i|^{1/3} \mathcal{T}^{2/3} (\log |\Sigma^i|)^{1/3}\Big); 
\end{aligned}
\end{equation*}
choosing $  
  \gamma : = 0 , \eta := \sqrt{\frac{\log |\Sigma^i|}{|\Sigma^i|\mathcal{T}}}$ yields
\begin{equation*}
\begin{aligned}
\mathbb{E}[\mathcal{R}^i_\mathcal{\mathrm{ext}} (\mathcal{T})]
  =
\mathcal{O} \Big(|\Sigma^i|^{1/2} \mathcal{T}^{1/2} (\log |\Sigma^i|)^{1/2}\Big).
\end{aligned}
\end{equation*}
\end{theorem}

We also prove that running Algorithm~\ref{algo:noregrettransmit} with Subroutine~\ref{algo:sub2_int} achieves sublinear internal regret, and thus yield convergence of empirical play to the CE set. The argument follows the standard reduction of swap-regret to a collection of external-regret problems on \emph{rows} of strategy candidates \cite{blum07externalinternal}, combined with a potential-function analysis for entropic OMD under importance-weighted bandit estimates.

\begin{theorem}\label{thm:swap_regret}
Fix a radar $i \in \Ical$ running Algorithm~\ref{algo:noregrettransmit} with time-invariant parameters $\eta_\tau \equiv \eta$ and $ \gamma_\tau \equiv \min \{ \gamma_{\max} , \gamma\}$, the expected cumulative internal regret for Subroutine~\ref{algo:sub2_int} satisfies
\begin{equation*}
\mathbb{E}\left[ \mathcal{R}^i_\mathcal{\mathrm{int}} (\mathcal{T})\right]
\ \le\
|\Sigma^i| \left(
\frac{\log |\Sigma^i|}{\eta}
+
\frac{\eta  \mathcal{T}|\Sigma^i|}{\gamma}
+
\gamma \mathcal{T}\right),
\end{equation*}
where the constant $\gamma_{\max} \in (0,1)$.

\noindent\textbf{Case $\gamma = 0$ (no exploration):} The bound simplifies to
\begin{equation*}
\mathbb{E} \left[ \mathcal{R}^i_{\mathrm{int}} (\mathcal{T})\right]
\le  |\Sigma^i| \left( \frac{\log |\Sigma^i|}{\eta} + \eta \mathcal{T}|\Sigma^i| \right)  .
\end{equation*}

The two choices of $\eta$ and $\gamma$ in Theorem~\ref{thm:bandit_ext_regret} yields
\begin{equation*}
\begin{aligned}
    \mathbb{E}[\mathcal{R}^i_\mathcal{\mathrm{int}} (\mathcal{T})]
&  =
\mathcal{O} \Big(|\Sigma^i|^{4/3} \mathcal{T}^{2/3} (\log |\Sigma^i|)^{1/3}\Big), \text{ and } \\
 \mathbb{E}[\mathcal{R}^i_\mathcal{\mathrm{int}} (\mathcal{T})]
 &  =
\mathcal{O} \Big(|\Sigma^i|^{3/2} \mathcal{T}^{1/2} (\log |\Sigma^i|)^{1/2}\Big), \text { respectively. }
\end{aligned}
\end{equation*}
\end{theorem}

We defer the proofs for Theorem~\ref{thm:bandit_ext_regret} and \ref{thm:swap_regret} to the Appendix.

\subsubsection{Implication for Convergence}

The regret bounds hence can be translated into finite–time equilibrium guarantees for the repeated game $\Gamma$ restricted to the finite strategy sets $\{\Sigma^i\}_{i\in\Ical}$ (i.e., the CPI–long schedules implementable by the stochastic round–robin sampler). 

\begin{corollary}\label{cor:CCE_CE_rate}
Let $\bar\pi_{\mathcal T}$ be the empirical joint-play distribution up to CPI $\mathcal T$ over the restricted finite game with strategy sets $\{\Sigma^i\}_{i\in\Ical}$, and let
\[
\varepsilon^{(i)}_{\mathrm{ext}}(\mathcal T)=\frac{1}{\mathcal T}  \mathbb{E}\big[\mathcal{R}^{i}_{\mathrm{ext}}(\mathcal T)\big],
\qquad
\varepsilon^{(i)}_{\mathrm{int}}(\mathcal T)=\frac{1}{\mathcal T}  \mathbb{E}\big[\mathcal{R}^{i}_{\mathrm{int}}(\mathcal T)\big],
\]
with $\varepsilon_{\mathrm{ext}}(\mathcal T):=\max_{i}\varepsilon^{(i)}_{\mathrm{ext}}(\mathcal T)$ and $\varepsilon_{\mathrm{int}}(\mathcal T):=\max_{i}\varepsilon^{(i)}_{\mathrm{int}}(\mathcal T)$. Then, 
 if each radar runs Subroutine~\ref{algo:sub2_ext} and satisfies Theorem~\ref{thm:bandit_ext_regret}  with  proper choices of  $\gamma \in (0, \gamma_{\max})$ and $\eta$, then for every player $i$ and every deviation $\hat\sigma^i\in\Sigma^i$, $\bar\pi_{\mathcal T}$ is an $\varepsilon_{\mathrm{ext}}(\mathcal T)$–CCE, where
$ \varepsilon_{\mathrm{ext}}(\mathcal T)=
\mathcal{O}\Bigg(\max_{i\in\Ical}|\Sigma^i|^{1/3} (\log |\Sigma^i|)^{1/3} \mathcal{T}^{-1/3}\Bigg)$.
(ii) If each radar runs Subroutine~\ref{algo:sub2_int} and satisfies Theorem~\ref{thm:swap_regret}  with proper choices of $\gamma \in (0, \gamma_{\max})$ and $\eta$, then for every player $i$ and every pair $\sigma^i,\sigma^{i\prime}\in\Sigma^i$, $\bar\pi_{\mathcal T}$ is an $\varepsilon_{\mathrm{int}}(\mathcal T)$–CE.
$\varepsilon_{\mathrm{int}}(\mathcal T)=
\mathcal{O}\Bigg(\max_{i\in\Ical}|\Sigma^i|^{4/3} (\log |\Sigma^i|)^{1/3} \mathcal{T}^{-1/3} \Bigg)$.
\end{corollary}

\begin{proof}[Proof sketch] Average the respective regret inequalities over $\tau=1{:}\mathcal T$ for each player, divide by $\mathcal T$, and rewrite the averages as expectations under $\bar\pi_{\mathcal T}$. Part (a) is the standard external–regret $\Rightarrow$ CCE implication (cf.\ Corollary~\ref{cor:cce}); part (b) follows from the swap–regret $\Rightarrow$ CE implication (cf.\ Lemma~\ref{lemma:cenoregret}), using pairwise swap maps and conditioning on the recommended action.
\end{proof}

% Alternative derivations replace (iii) with \emph{Blackwell approachability} of the negative orthant for the vector of pairwise regrets, or with \emph{regret matching$+$} potentials; both yield the same $O(M\sqrt{T\log M})$ rate under bandit feedback and bounded rewards. Either framework implies that vanishing swap regret makes the empirical play an approximate CE of the underlying finite game.

\section{Numerical Experiments}
\label{sec:experiment}

In this section, we validate the proposed No-Regret Hopping framework through comprehensive numerical simulations. We consider a challenging interference scenario and compare our approach against several baseline methods.

\subsection{Example Scenario: Single-Target}

\subsubsection{Scenario Configuration}

We consider a scenario where four FMCW radars operate in close proximity. The radars are positioned randomly within a circular region of radius 25 meters around a common target located at the origin, as illustrated in Fig.~\ref{fig:radar_config}.
\begin{figure}[t]
    \centering
    \includegraphics[width=0.9\linewidth]{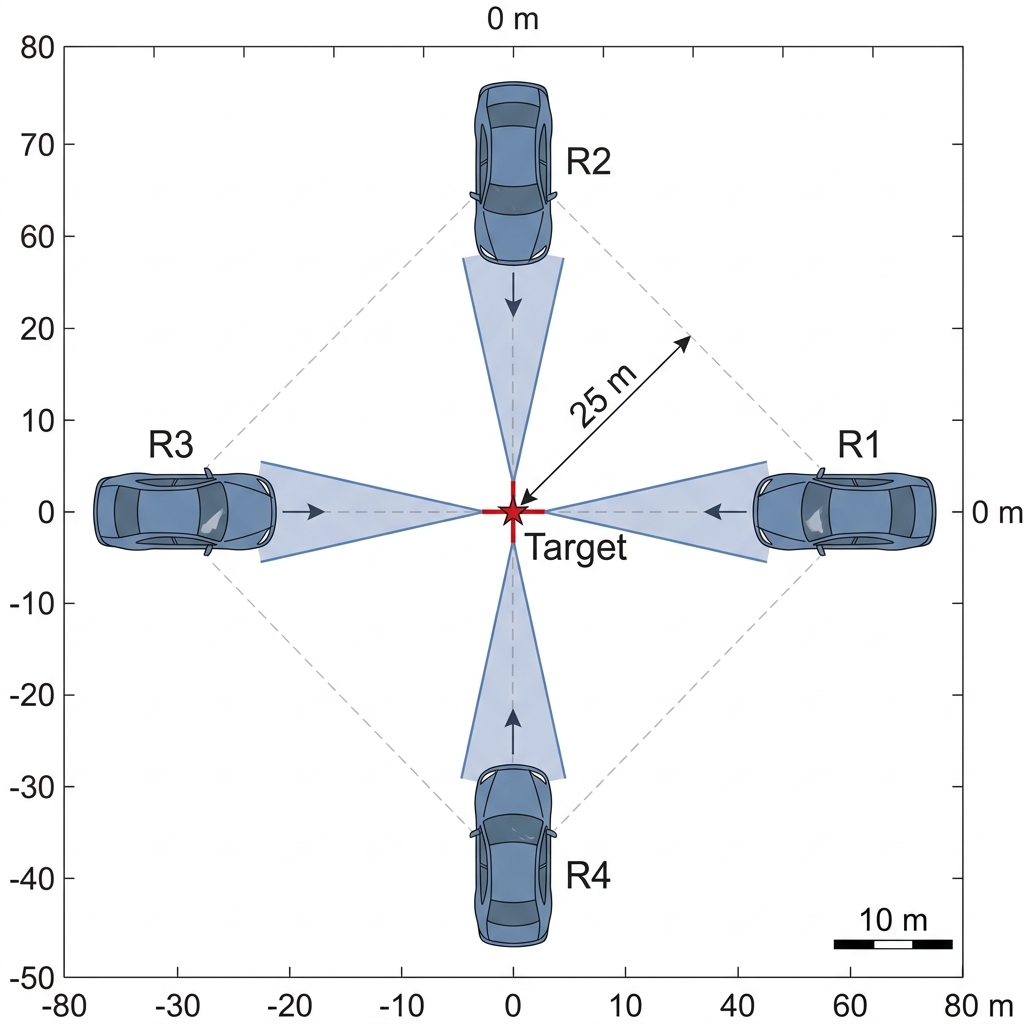}
    \caption{Example scenario: $I=4$ radars with  fully connected interference graph and synchronized operations.}
    \label{fig:radar_config}
\end{figure}
We assume the interference graph $\mathcal{G}$ is fully connected, meaning each radar potentially interferes with all others. Furthermore, we assume perfect synchronization at initialization, i.e., $\Delta_{io} = 0$ for all $(i \to o) \in \mathcal{E}$, representing a worst-case scenario encountered in dense traffic conditions. This synchronized worst-case initialization tests the algorithms' ability to learn effective interference avoidance strategies from scratch.

The radar parameters are configured as follows:
Each radar operates at a 77 GHz starting carrier frequency $f_c$. The available spectrum is divided into $A_f = 3$ frequency subbands, each with nominal bandwidth $B = 150$ MHz. 
The idle interval $[T^i_{\text{PRI}} - T^i_a - T_{\text{reset}}]$ for each radar is discretized into $A_t$ time slots, where $T_{\text{reset}}$ denotes the PLL settling time required after a frequency hop. The time-shifting action set is defined as $\mathcal{T}_h = \{0, 2T_{\text{sett}}, 4T_{\text{sett}}, \ldots, T_{\text{reset}}\}$ with $T_{\text{sett}} = 1.5$ $\mu$s serving as a guard interval. This yields approximately $A_t = 7$ discrete time-shifting slots per chirp. Each radar transmits $K = 256$ chirps per CPI with pulse repetition interval $T_{\text{PRI}} = 29.99$ $\mu$s and active chirp duration $T_a = 8.89$ $\mu$s.
The radars are uniformly distributed within $5$ m circles around the vertices of the polygon with radius $25$ m, detecting with relative velocity $v$ uniformly distributed within the interval of $[-25, 25]$ m/s. The target is placed uniformly randomly within the circle around the origin.  The radar cross-section of the target is $\sigma_{\text{RCS}} = 20$ dBsm. Individual radar bandwidths vary slightly within $[110, 150]$ MHz to emulate realistic heterogeneity in automotive radar systems.
Table~\ref{tab:radar_params} summarizes the key radar parameters used in the simulations. 
All simulations use a common random seed for reproducibility.

\begin{table}[t]
\centering
\caption{Radar System Parameters}
\label{tab:radar_params}
\begin{tabular}{lc}
\hline
\textbf{Parameter} & \textbf{Value} \\
\hline
Frequency subbands $A_f$ $\times$ Time slots $A_t$ & 3 $\times$ 7 \\
Carrier frequency $f_c$ (GHz) & 77 \\
Subband bandwidth $B$ (MHz) & uniform[110, 150] \\
Active duration $T_a$ ($\mu$s) & 8.89 \\
Pulse repetition interval $T_{\text{PRI}}$ ($\mu$s) & 29.99 \\
Chirps per CPI $K$ & 256 \\
Target range $r$ (m) & 25 + uniform[-5,5] \\
Target velocity $v$ (m/s) &  uniform[-25, 25]  \\
TX power $P_t$ (dBm) & 13 \\
Max unambiguous range $R_u$ (m) & 200 \\
Noise power $P_N$ (dBm) & $-88$ \\
Block length $\ell$ (chirps) & $\max(A_f, A_t) = 7$ \\
\hline
\end{tabular}
\end{table}

\subsubsection{Algorithm Evaluation Setup}

We evaluate four variants of Algorithm~\ref{algo:noregrettransmit}: Random Hopping (Random), Nash Equilibrium (Nash), External Regret Minimization (External), and Internal Regret Minimization (Internal). In the Random and Nash variants, the strategy-update mechanism is disabled, whereas the External and Internal variants employ Subroutines~\ref{algo:sub2_ext} and~\ref{algo:sub2_int}, respectively, for online strategy updates. To ensure a direct correspondence between mixed strategies and joint time–frequency actions, we construct each radar’s strategy set $\Sigma^i$ to have the same cardinality as the joint action space $\mathcal{A}_f \times \mathcal{A}_t$, with each strategy uniquely associated with a specific $(f,t)$ pair. Under this design, sampling from $q^i_\tau(\cdot)$ is equivalent to sampling from $p^i_\tau(\cdot)$, enabling a unified treatment of pure strategies across all algorithmic variants.

\begin{enumerate}
    \item \textbf{Random}: Each radar adopts $p^i_\tau = \mathrm{Unif}(\Sigma^i)$, uniformly samples starting $f$-$t$ pairs for Subroutine~\ref{algo:sub1}.
    
    \item \textbf{Nash}: Each radar is assigned a unique starting frequency-time action pair deterministically throughout the blocks of Subroutine~\ref{algo:sub1}. We assign radar $i$ with joint action $\text{mod}(i-1, A_f \times A_t) + 1$. This achieves perfect coordination but relies on centralized scheduling.
    
    % \item \textbf{Nash Equilibrium (Single-band)}: A variant where radars use predetermined frequency assignments without time-shifting ($A_t = 1$), representing the traditional frequency division approach.
    
    \item \textbf{External}: 
    We use a temperature-tampered constant learning rate $\eta_\tau = \kappa  \frac{ \log(21)^{2/3}}{ 21^{1/3} \times 15^{2/3}} = 0.1252$ with the scaling constant $\kappa = 36$. Exploration rate $\gamma_\tau$ is set to be linearly decayed from $ 0.1 $ to $0$.
    % Each radar minimizes external regret using algorithm with temperature parameter $\tau = 36$ and learning rate $\eta_t = \sqrt{2\log(A)/T}$ where $A = A_f \times A_t$ is the action space size.
    
    \item \textbf{Internal}: 
     We set the exploration rate $\gamma_\tau \equiv 0$, and the constant learning rate for updating regret matrices to be $\eta_\tau = 0.5 $. We apply a positive-part thresholding operator $[\cdot]_+$ to the score matrix $z^i_\tau$ prior to the softmax mapping at each CPI $\tau$, in direct analogy with the Hart-Mas-Colell regret-matching construction~\cite{hart2000simple}. 

    % The correlated equilibrium approach using Hart and Mas-Colell's internal regret minimization. Each radar maintains a transition matrix and updates according to the stationary distribution with temperature parameter $\tau = 0.5$ for intra-group action selection.
\end{enumerate}

We run 15 epochs (CPIs) for each aforementioned method, and record the last epoch performance. For learning methods, we also record the convergence behavior. 
The following metrics are used for the evaluation:
\begin{itemize}
    \item SINR performance: Measured in dB, averaging the chirps of the last CPI, then averaging across the radars.
    \item Collision Rate: The number of frequency-time slots where two or more radars transmit simultaneously, normalized by the total number of slots.
    \item Convergence Behavior: Evolution of mixed strategies over epochs for the learning methods, visualized through strategy probability distributions. We show only the top $8$ probability masses over the $f$-$t$ pairs for clarity.
    \item {Range-Doppler (RD) Map Quality}: Visual assessment of target detection clarity in the 3D RD map output, with targets marked at their ground-truth positions.
\end{itemize}

\subsection{Example Scenario Results}

\subsubsection{Frequency-Time Hopping Illustration}
Fig.~\ref{fig:chirp_distribution} illustrates the hopping pattern for different methods through a frequency-time diagram snapshot between intermediate chirps. Over the full experiment, the random baseline exhibits uniform but highly colliding patterns, with approximately 35\% of time-frequency slots experiencing collisions. The Nash joint method achieves perfect separation with zero collisions through predetermined assignments, as each radar occupies a unique frequency-time cell. Our proposed methods demonstrate intermediate behavior. The external regret minimization method learns to avoid high-collision regions, achieving approximately 3\% collision rate after convergence. The internal regret minimization approach exhibits even better coordination, reaching 0\% collision rate, because the method is exploration-free and the computation of the stationary distribution itself provides sufficient implicit exploration.

%Fig.~\ref{fig:chirp_distribution} illustrates the frequency-time utilization patterns for different methods during the final CPI. The random baseline (Fig.~\ref{fig:random_illu}) exhibits uniform but highly colliding patterns, with approximately 67\% of time-frequency slots experiencing collisions. The Nash joint method (Fig.~\ref{fig:nash_illu}) achieves perfect separation with zero collisions through predetermined assignments, as each radar occupies a unique frequency-time cell. The Nash single-band variant (Fig.~\ref{fig:nashsingle_illu}) shows perfect frequency separation but no time diversity.

% Our proposed methods demonstrate intermediate behavior. The external regret minimization method (Fig.~\ref{fig:ext_illu}) learns to avoid high-collision regions, achieving approximately 15\% collision rate after convergence. The internal regret minimization approach (Fig.~\ref{fig:internal_illu}) exhibits even better coordination, reaching approximately 8\% collision rate, demonstrating the benefit of correlation in action selection.

\begin{figure}[t]
    \centering
    \includegraphics[width=\linewidth]{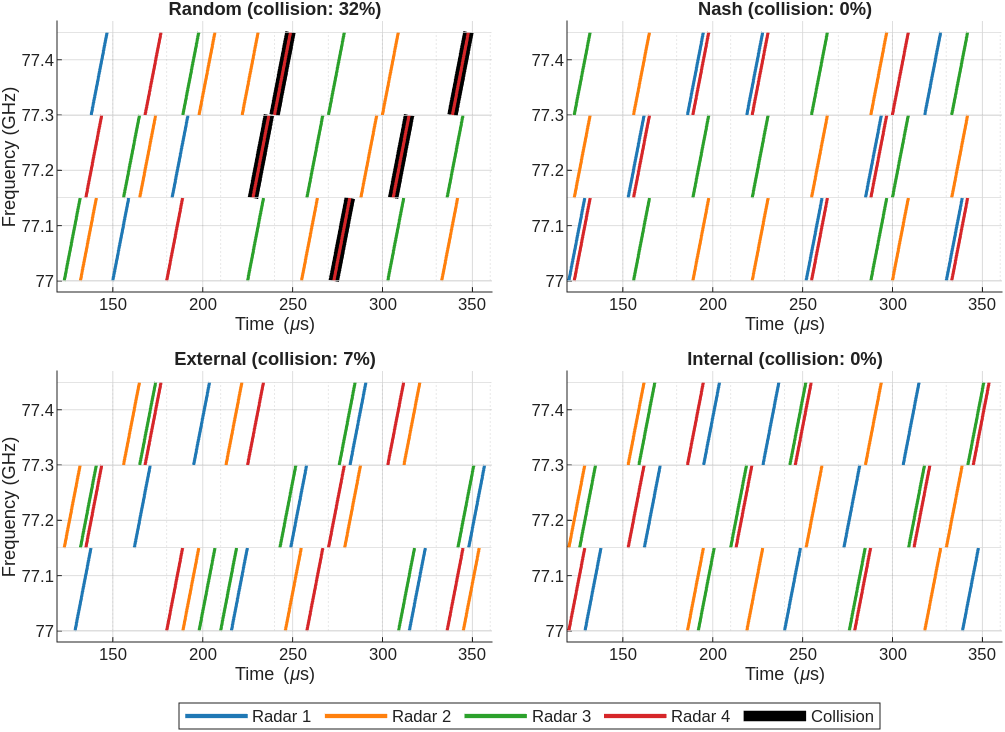}
    \caption{Frequency-time chirp illustration for different scheduling methods in the final CPI. Each panel shows LFM chirp ramps across several PRIs for four radars. Collisions (same frequency-time slot occupied by multiple interfering radars) are highlighted with thick black bands. Collision rates are computed over the full CPI ($K=256$ chirps).}
    \label{fig:chirp_distribution}
\end{figure}

% \begin{figure*}[ht]
% \centering
% \begin{subfigure}[b]{0.48\textwidth}
%     \centering
%     \includegraphics[width=\textwidth]{figs/nash/nash_collision.png}
%     \caption{Nash equilibrium (joint)}
%     \label{fig:nash_collision}
% \end{subfigure}
% \begin{subfigure}[b]{0.48\textwidth}
%     \centering
%     \includegraphics[width=\textwidth]{figs/random/random_collision.png}
%     \caption{Random hopping}
%     \label{fig:random_collision}
% \end{subfigure}
% \begin{subfigure}[b]{0.48\textwidth}
%     \centering
%     \includegraphics[width=\textwidth]{figs/ext/ext_collision.png}
%     \caption{External regret minimization}
%     \label{fig:ext_collision}
% \end{subfigure}
% \hfill
% \begin{subfigure}[b]{0.48\textwidth}
%     \centering
%     \includegraphics[width=\textwidth]{figs/int/internal_collision.png}
%     \caption{Internal regret minimization}
%     \label{fig:internal_collision}
% \end{subfigure}
% \hfill

% \caption{Collision heatmaps showing which frequency-time cells experience interference from multiple radars. White cells indicate no collision, while red cells indicate two or more radars selecting the same action.}
% \label{fig:collision_maps}
% \end{figure*}

\subsubsection{Strategy Convergence}

Fig.~\ref{fig:convergence} presents the evolution of mixed strategies for individual radars under the learning-based methods. Each radar's strategy is represented as a probability distribution over the $A_f \times A_t = 21$ joint actions.

\begin{figure*}[ht]
\centering
\begin{subfigure}[b]{0.485\textwidth}
    \centering
    \includegraphics[width=\linewidth]{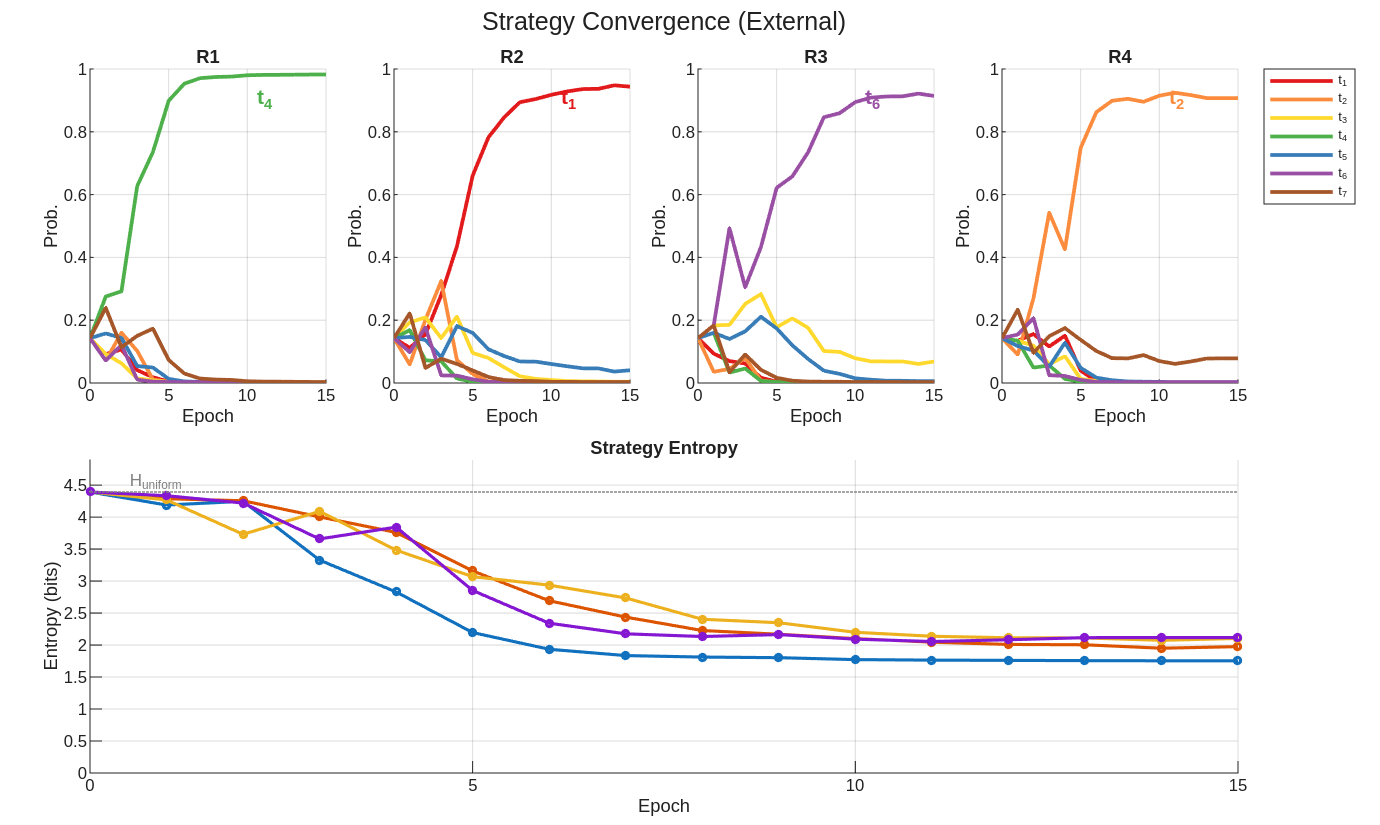}
\end{subfigure}
\hfill
\begin{subfigure}[b]{0.485\textwidth}
    \centering
    \includegraphics[width=\linewidth]{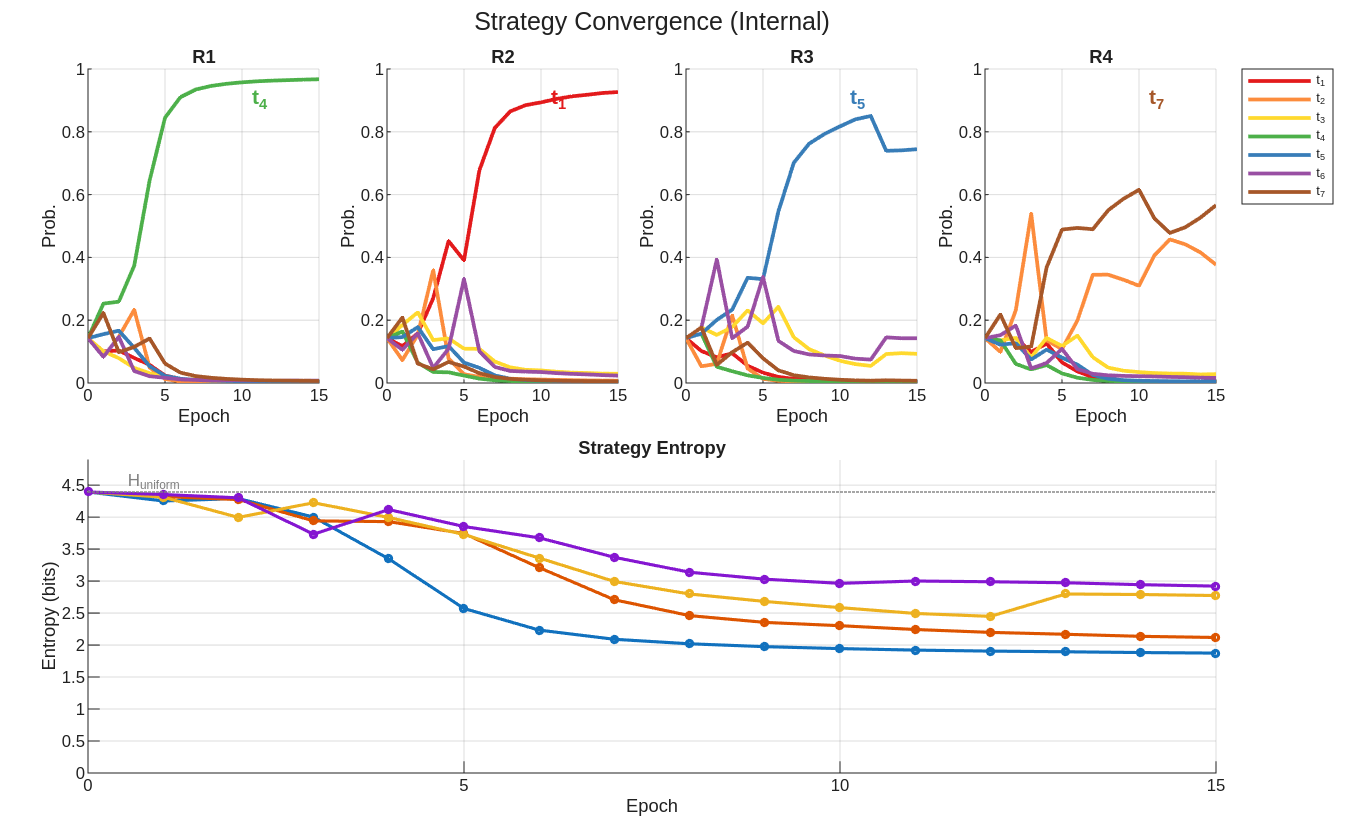}
\end{subfigure}
\caption{Strategy convergence under external (left) and internal (right) regret minimization. Top panels show per-radar group (time-slot) probability evolution; each colored curve represents one of $A_t = 7$ time-slot groups. The dominant group at the final epoch is annotated. Bottom panels show the strategy entropy decreasing from the uniform level toward concentration. Different radars converge to different dominant groups, achieving anti-coordination.}
\label{fig:convergence}
\end{figure*}

The external regret method concentrates each radar's strategy onto a distinct time-slot group within the first 5--8 epochs, as shown by the entropy decrease in Fig.~\ref{fig:convergence} (left). By epoch~10, each radar assigns the majority of its probability mass to a single group, with small residual mass on the remaining groups providing spectral diversity. The internal regret method exhibits a similar convergence pattern, driven by the cumulative swap-regret scores that penalize groups where collisions occur. Both methods achieve anti-coordination: the annotated dominant groups differ across radars, ensuring that interfering pairs occupy distinct time slots.

\subsubsection{Final CPI Results}

\begin{figure*}[ht]
\centering
\begin{subfigure}[b]{0.487\textwidth}
    \centering
    \includegraphics[width=\textwidth]{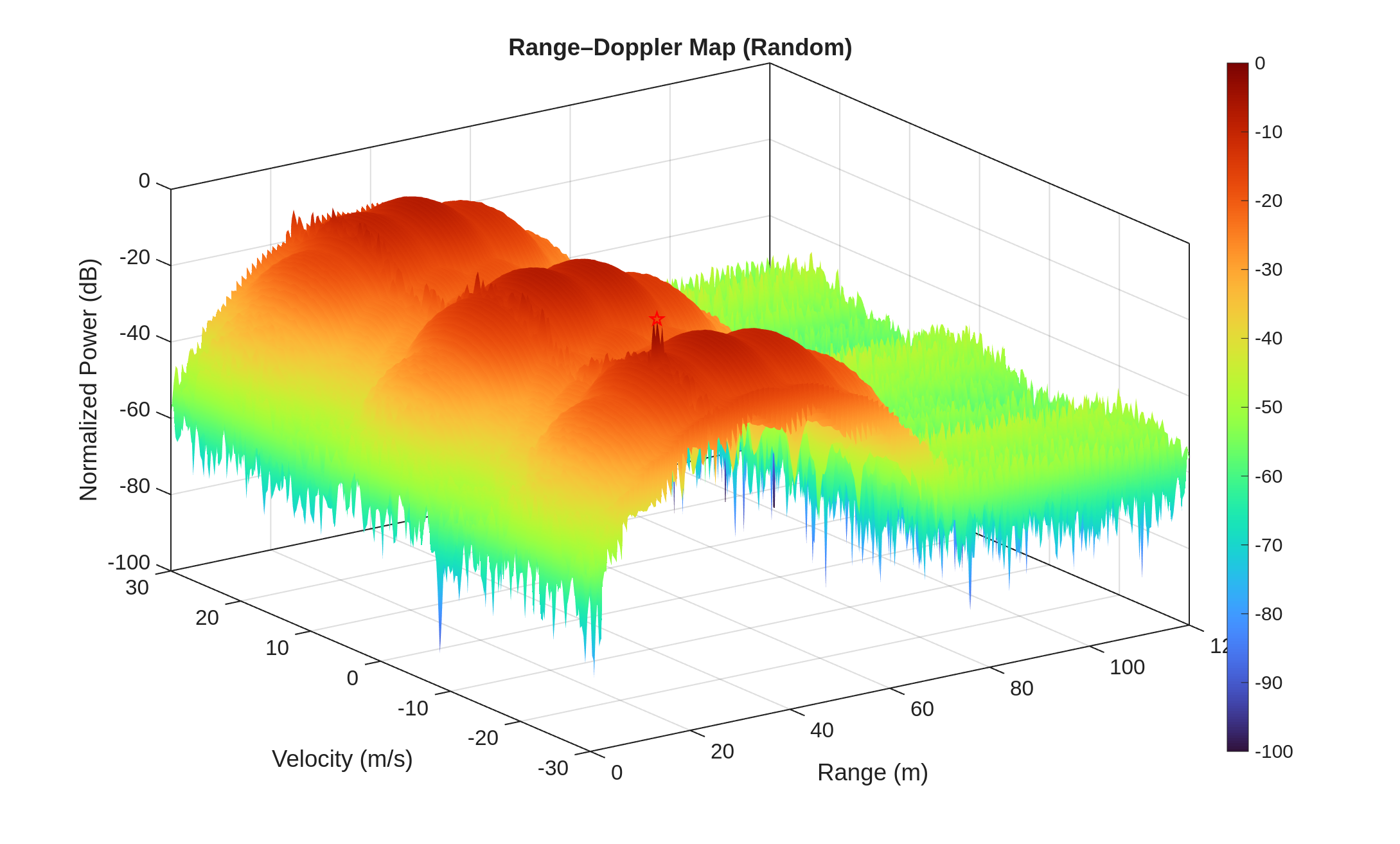}
  
\end{subfigure}
\hfill
\begin{subfigure}[b]{0.487\textwidth}
    \centering
    \includegraphics[width=\textwidth]{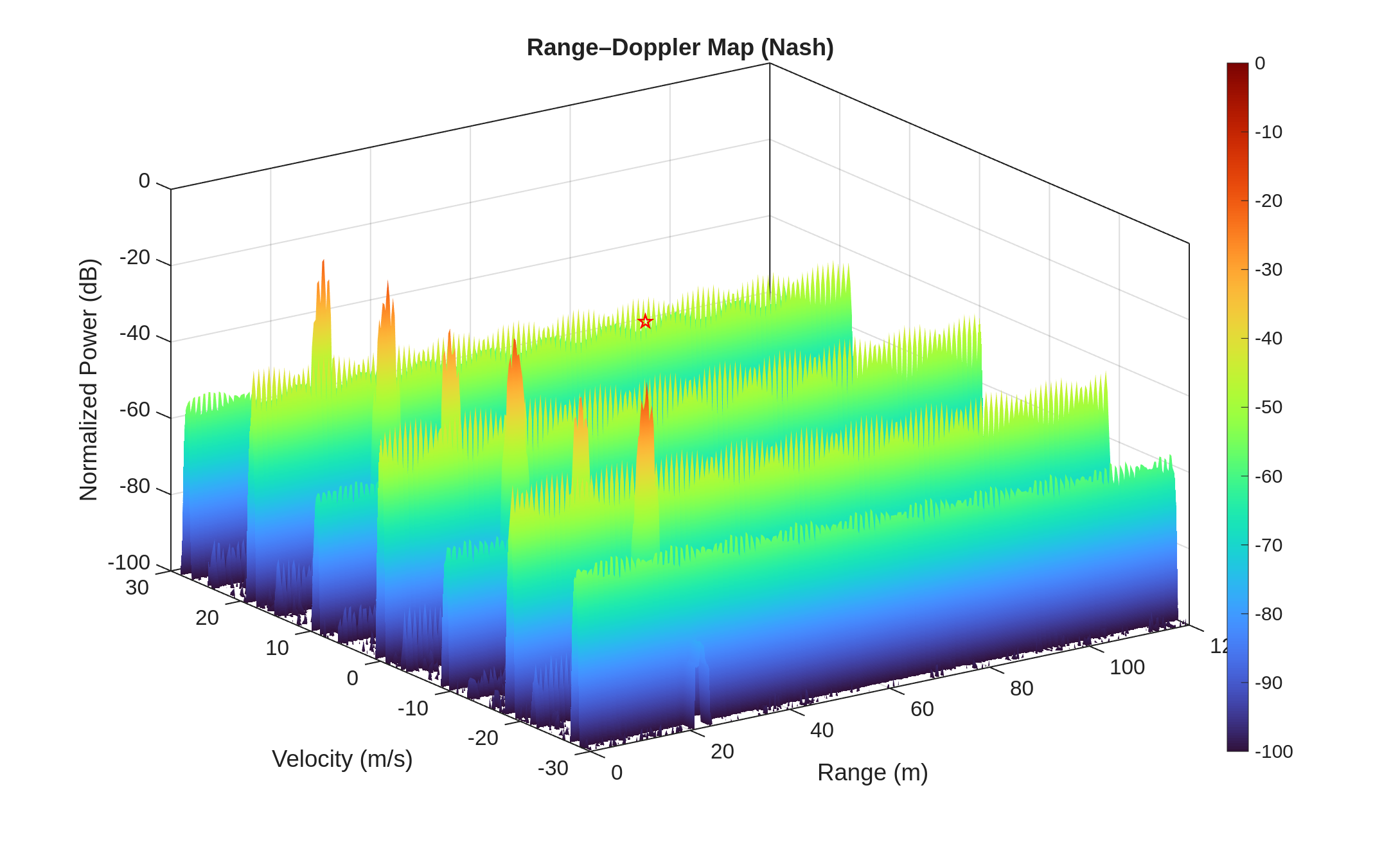}
   
\end{subfigure}

\begin{subfigure}[b]{0.487\textwidth}
    \centering
    \includegraphics[width=\textwidth]{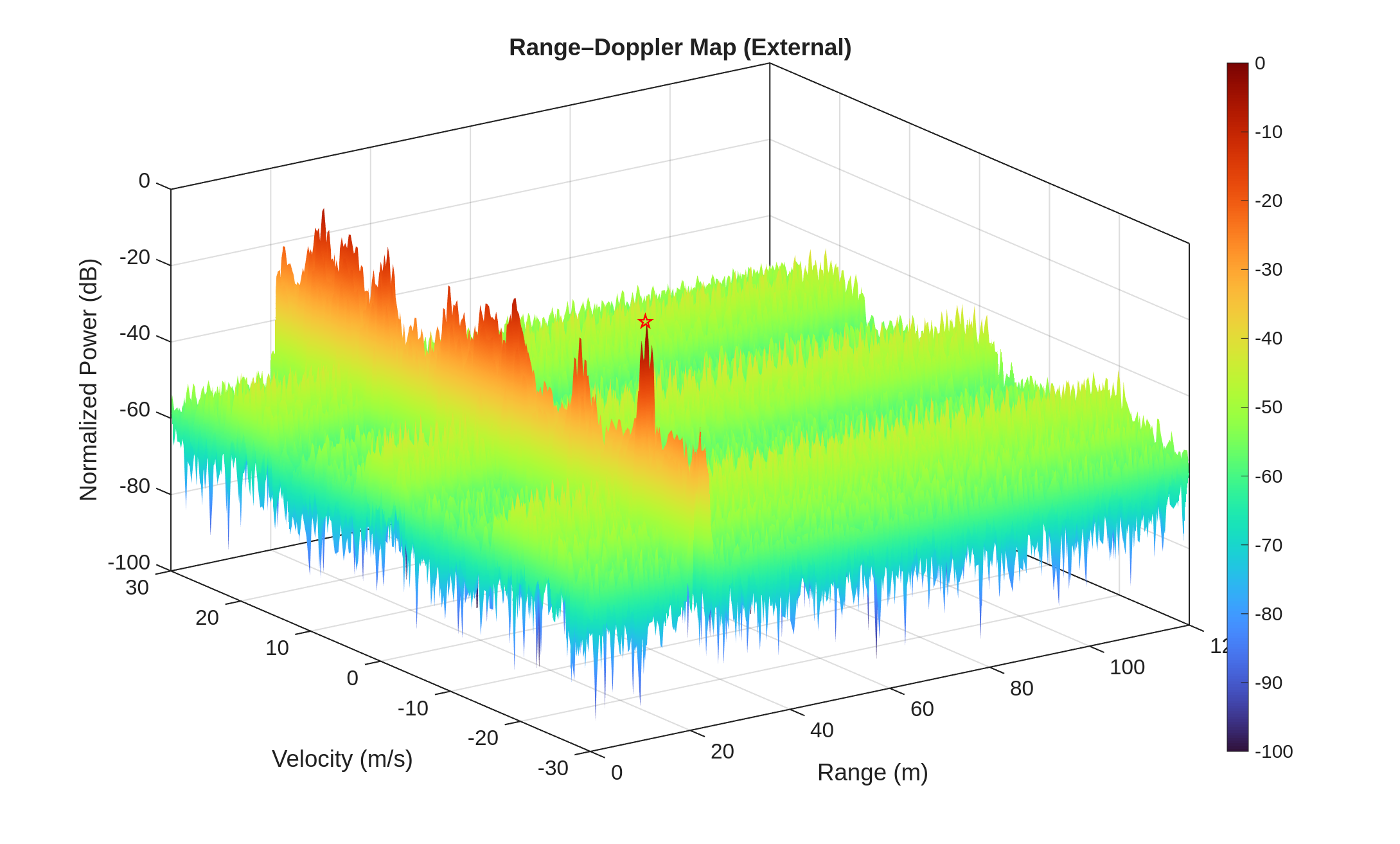}
    
\end{subfigure}
\hfill
\begin{subfigure}[b]{0.487\textwidth}
    \centering
    \includegraphics[width=\textwidth]{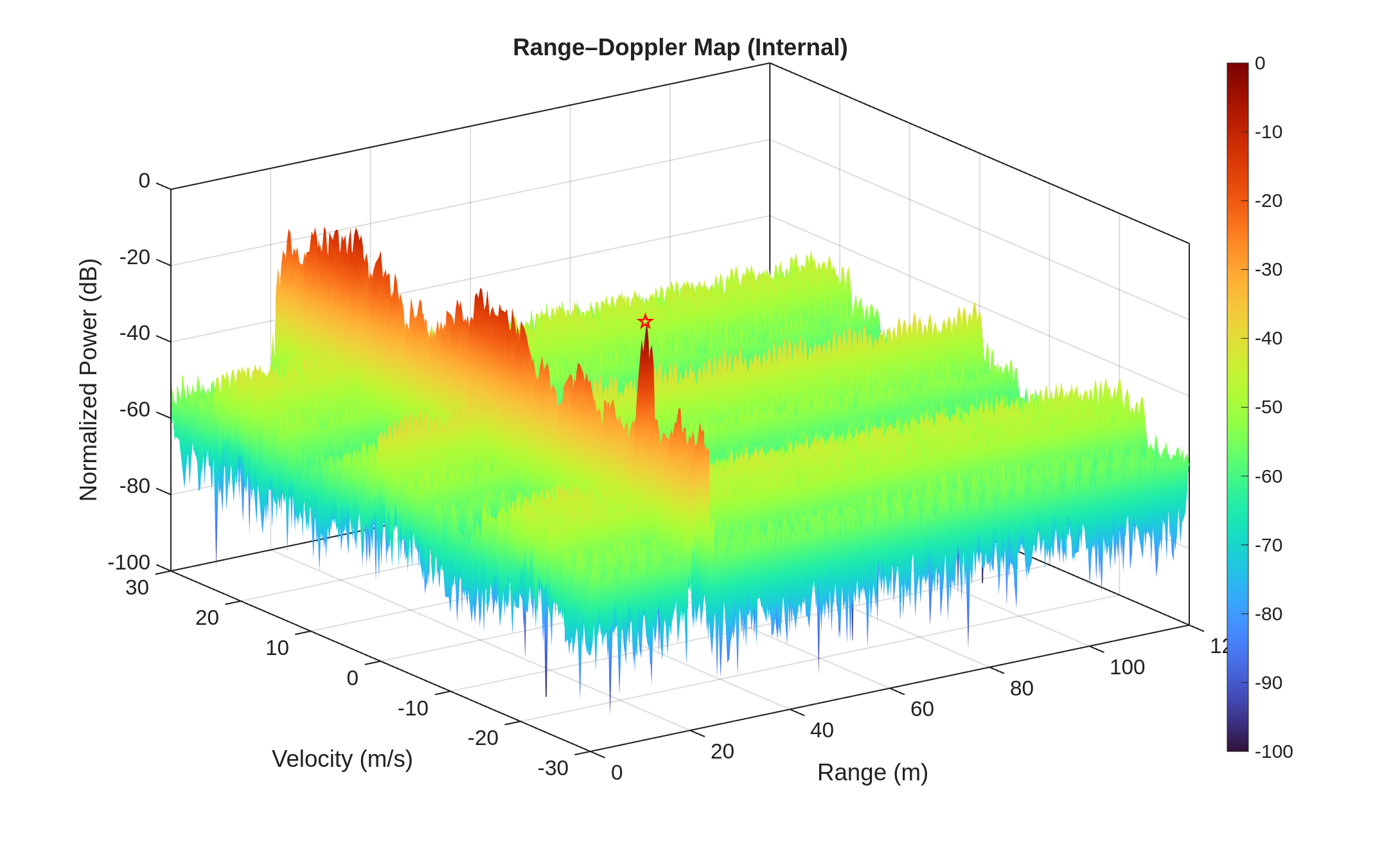}

\end{subfigure}
\caption{The final CPI RD maps for an individual radar operating with different methods.}
\label{fig:rdmap}
\end{figure*}

% \begin{figure}[htp]
%     \centering
%     \includegraphics[width=\linewidth]{figs/collision_rate_bars.png}
%     \caption{The collision rate bar charts for the final epoch (CPI) across different methods.}
%     \label{fig:collision_rate}
% \end{figure}

Fig.~\ref{fig:rdmap} presents the final converged RD maps for a representative radar under the different scheduling methods. The results clearly illustrate the distinctive behaviors of each approach. The random method fails to avoid interference, leading to a significant elevation of the noise floor across the RD map and causing the target response to be entirely submerged. The Nash method, although theoretically free of collisions, exhibits noticeable Doppler aliasing. This artifact arises from the deterministic periodic frequency assignment produced by its round-robin strategy (Algorithm~\ref{algo:sub1}). We note that a centrally coordinated scheme could, in principle, design a non-periodic hopping pattern to mitigate this aliasing; however, such designs require global knowledge of all radars' parameters and explicit scheduling, which is impractical in decentralized automotive settings. In contrast, the regret-based methods deliver substantially improved performance without any central coordination. Both external and internal regret minimization successfully learn interference-avoiding patterns.

Regarding collision statistics, the Nash method attains a collision rate of zero, consistent with its design. The two regret-based methods progressively reduce their collision rates over training, with the internal regret method driving the rate to zero within approximately 12 epochs. The random method, by comparison, maintains a persistently high collision rate due to its independent uniform sampling of time–frequency actions.

\subsubsection{SINR Performance}
We extended the radar configuration to scenarios of 3-7 radars, and performed 50 Monte-Carlo (MC) trials for each algorithm in each scenario, recording the average SINR across radars and the standard deviation across trials. The results are plotted in Fig.~\ref{fig:sinr_comparison}.
\begin{figure}[htbp]
    \centering
    \includegraphics[width=\linewidth]{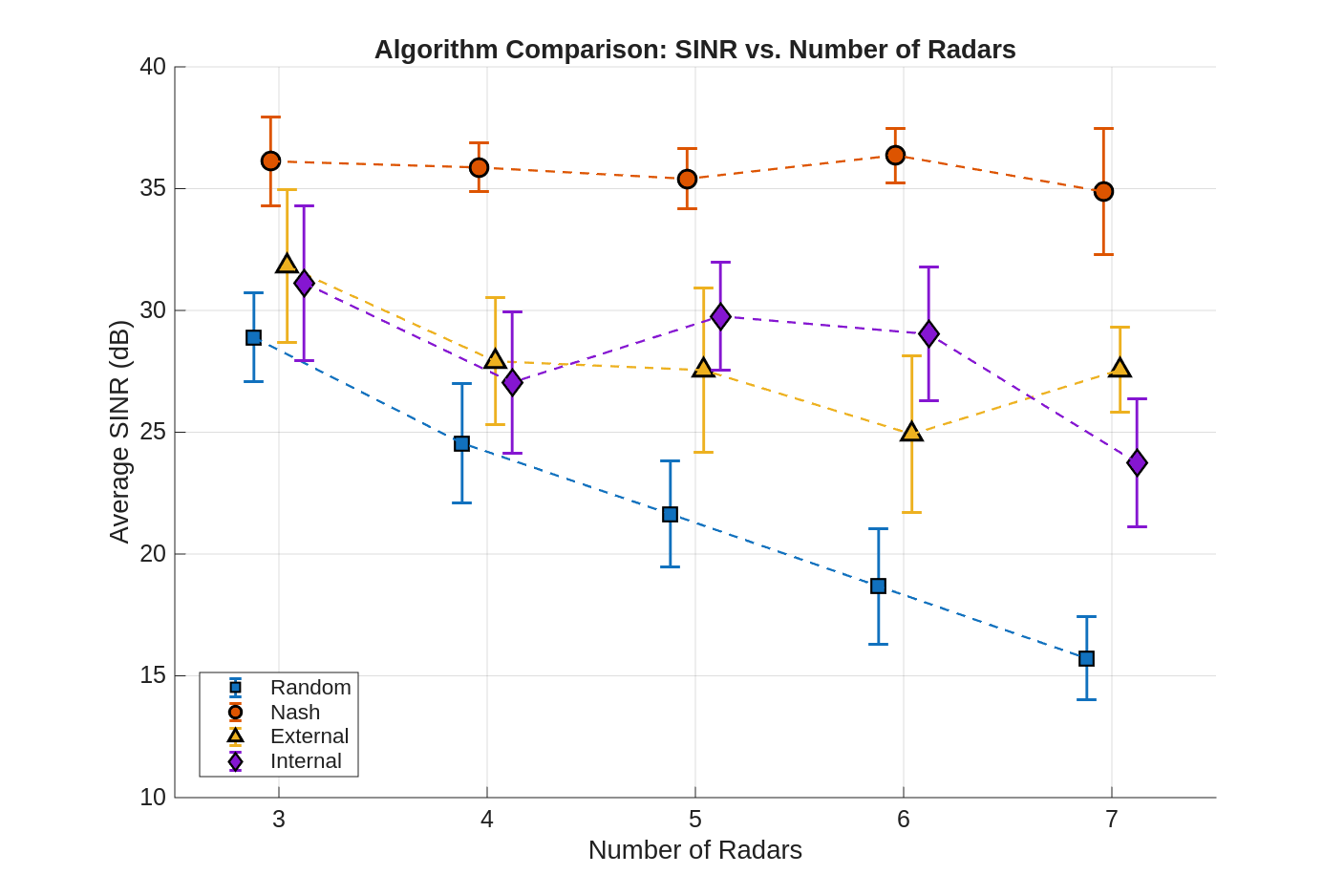}
    \caption{Average SINR comparison for all methods. The error bars represent the standard deviation of the per-trial mean SINR, computed across all radars.}    \label{fig:sinr_comparison}
\end{figure}

We observe that the random baseline produces the lowest SINR (approx. 28, 25, 22, 19, and 16 dB for 3--7 radars), a result inherent to its time-shifting design. The Nash joint method achieves the highest per-chirp SINR due to its deterministic collision-free assignment; however, as discussed above, it introduces Doppler aliasing artifacts that degrade range-Doppler map quality. The regret-based methods achieve competitive SINR that is slightly below Nash in the per-chirp metric, while avoiding aliasing through their stochastic hopping patterns. Between the two regret minimization approaches, no explicit performance advantage was evident.

Overall, both regret-based methods reliably learn collision-free time–frequency patterns and deliver competitive SINR and clearer RD maps than the randomized baseline, while matching Nash's collision-free performance without requiring any centralized coordination, demonstrating the effectiveness of joint time–frequency no-regret scheduling in dense multi-radar environments.

\subsection{Dynamic Traffic Scenarios}
\label{sec:dynamic_scenarios}

To evaluate the framework under realistic automotive conditions, we consider two dynamic traffic scenarios: an \emph{urban intersection} and a \emph{highway} with oncoming traffic. Both scenarios feature time-varying interference graphs, multi-target detection (each radar detects the nearest vehicle), and dynamic radar entry, testing the framework well beyond the static 4-radar setting above.

\subsubsection{Scenario configuration}
Fig.~\ref{fig:scenario_geometry} illustrates both scenarios. Five vehicles equipped with front-facing 77~GHz FMCW radars approach from the four cardinal directions at 30--50~km/h; one additional vehicle enters from a southwest side street at epoch~15 (of 30 total) with a cold-start uniform strategy. In the highway scenario, inspired by the traffic model of Cuccoli et al.~\cite{cuccoli2025highway}, six vehicles share a multi-lane road: four eastbound (100--130~km/h) and two oncoming westbound (100--120~km/h), producing closing speeds up to 260~km/h. Two new vehicles (an overtaker and an on-ramp merger) enter at epoch~15. Table~\ref{tab:dynamic_params} summarizes the parameters. In both scenarios, the interference graph is recomputed each CPI based on a 200~m distance threshold, and each radar detects the nearest other vehicle as its primary target.

\begin{figure*}[t]
\centering
\begin{subfigure}[b]{0.48\textwidth}
    \centering
    \includegraphics[width=\linewidth]{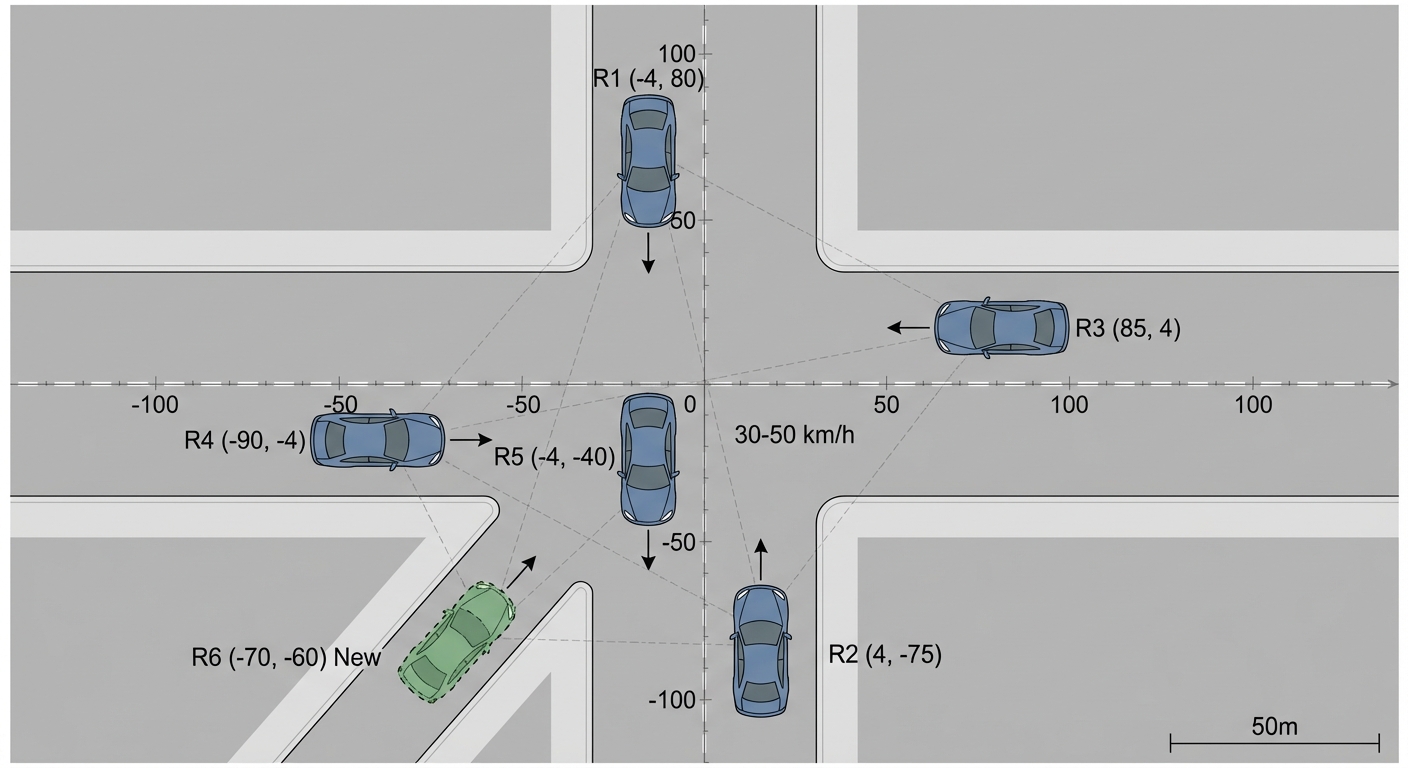}
    \caption{Urban intersection}
    \label{fig:urban_geometry}
\end{subfigure}
\hfill
\begin{subfigure}[b]{0.48\textwidth}
    \centering
    \includegraphics[width=\linewidth]{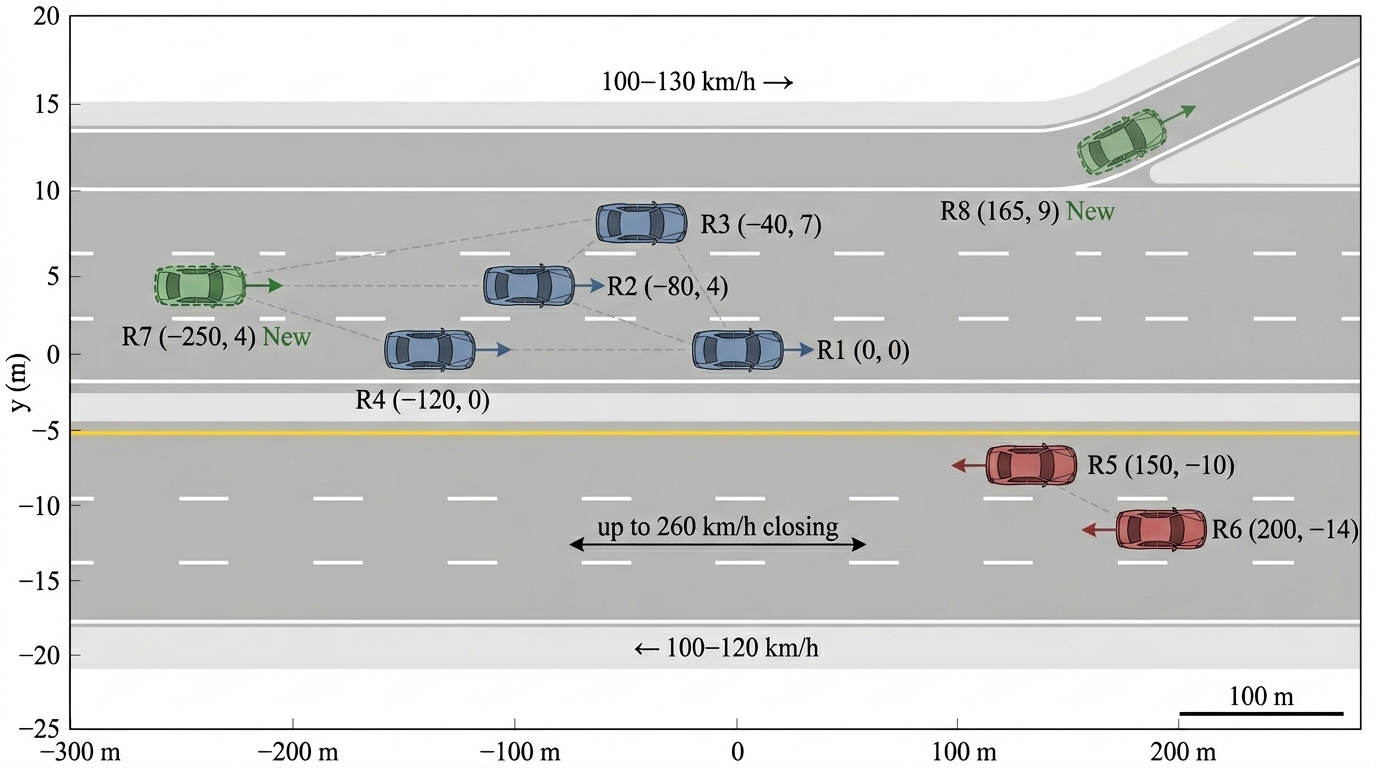}
    \caption{Highway}
    \label{fig:highway_geometry}
\end{subfigure}
\caption{Dynamic traffic scenarios. (a)~Urban intersection: five vehicles (blue) approach from four directions; a new vehicle (green, dashed) enters from the southwest at epoch~15. (b)~Highway: four eastbound (blue) and two oncoming westbound (red) vehicles; two new vehicles (green, dashed) enter at epoch~15. Dashed gray lines indicate interference links (distance $\leq 200$\,m).}
\label{fig:scenario_geometry}
\end{figure*}

\begin{table}[t]
\centering
\caption{Dynamic Scenario Parameters}
\label{tab:dynamic_params}
\begin{tabular}{lcc}
\hline
\textbf{Parameter} & \textbf{Urban} & \textbf{Highway} \\
\hline
Initial / final radars & 5 / 6 & 6 / 8 \\
Vehicle speeds (km/h) & 30--50 & 100--130 / 100--120 \\
Max closing speed (km/h) & $\approx$50 & $\approx$260 \\
Inter-vehicle ranges (m) & 35--175 & 40--300 \\
Interference cutoff (m) & \multicolumn{2}{c}{200} \\
New radar entry epoch & \multicolumn{2}{c}{15 (of 30)} \\
Total epochs (CPIs) & \multicolumn{2}{c}{30} \\
CPI duration (ms) & \multicolumn{2}{c}{$\approx$7.68} \\
Target RCS (dBsm) & \multicolumn{2}{c}{20} \\
TX power $P_t$ / Gain $G^2$ & \multicolumn{2}{c}{23\,dBm / 46\,dBi} \\
\hline
\end{tabular}
\end{table}

\subsubsection{Collision rate and SINR convergence}
Fig.~\ref{fig:dynamic_results} presents the interference-aware collision rate and mean SINR across all 30~epochs for both scenarios. The External and Internal methods both drive collision rates toward zero within the first 10--12~epochs (Fig.~\ref{fig:dynamic_collision}). At epoch~15, the entry of new radars causes a transient spike, which is resolved within 3--5~epochs as the new radars learn to anti-coordinate with their neighbors. The mean SINR (Fig.~\ref{fig:dynamic_sinr}) improves correspondingly as collision rates decrease, improving toward the interference-free baseline. In the highway scenario, the sparser interference graph leads to higher baseline SINR, while the urban scenario's denser connectivity presents a harder coordination challenge.

\begin{figure*}[t]
\centering
\begin{subfigure}[b]{0.48\textwidth}
    \centering
    \includegraphics[width=\linewidth]{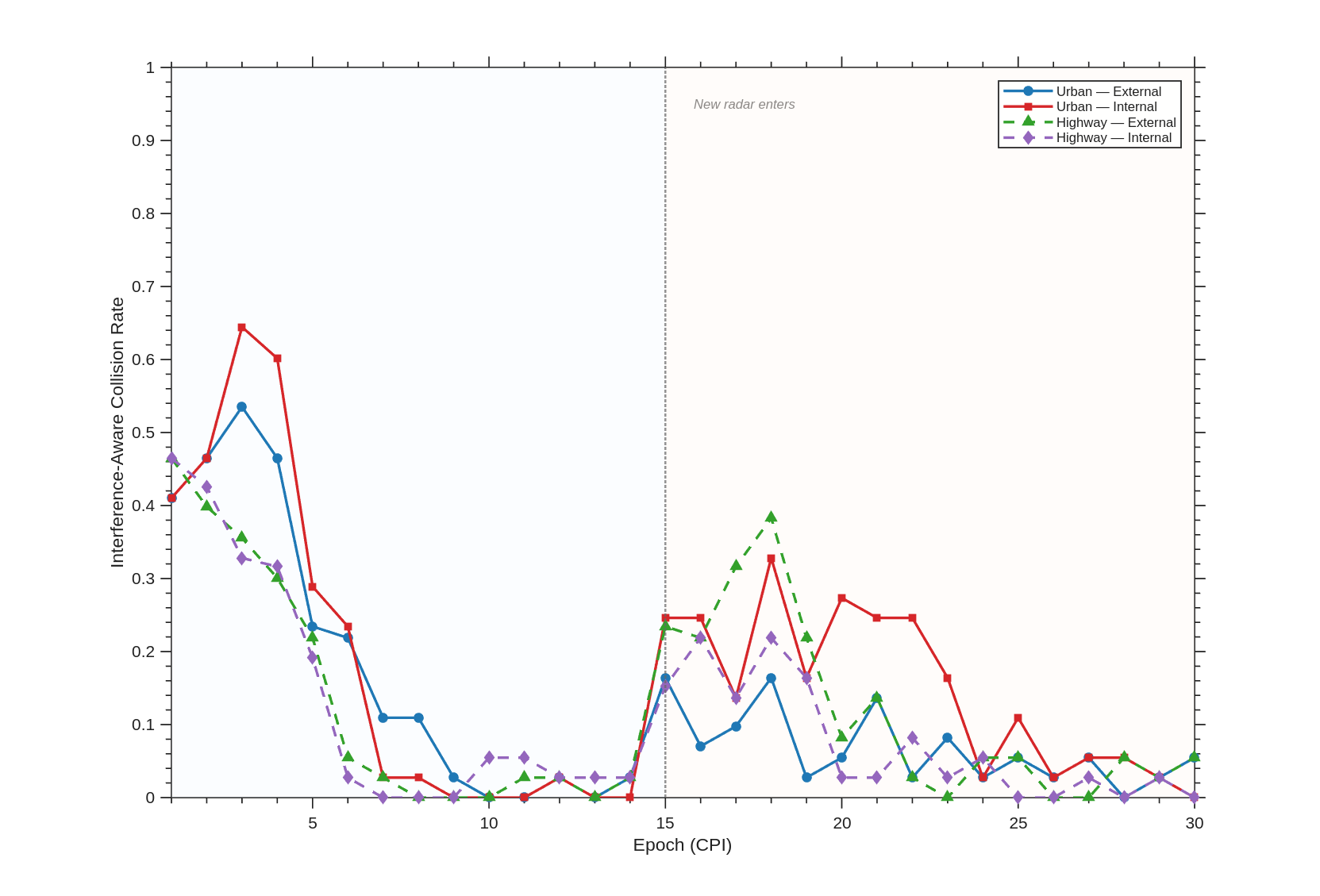}
    \caption{Collision rate vs.\ epoch}
    \label{fig:dynamic_collision}
\end{subfigure}
\hfill
\begin{subfigure}[b]{0.48\textwidth}
    \centering
    \includegraphics[width=\linewidth]{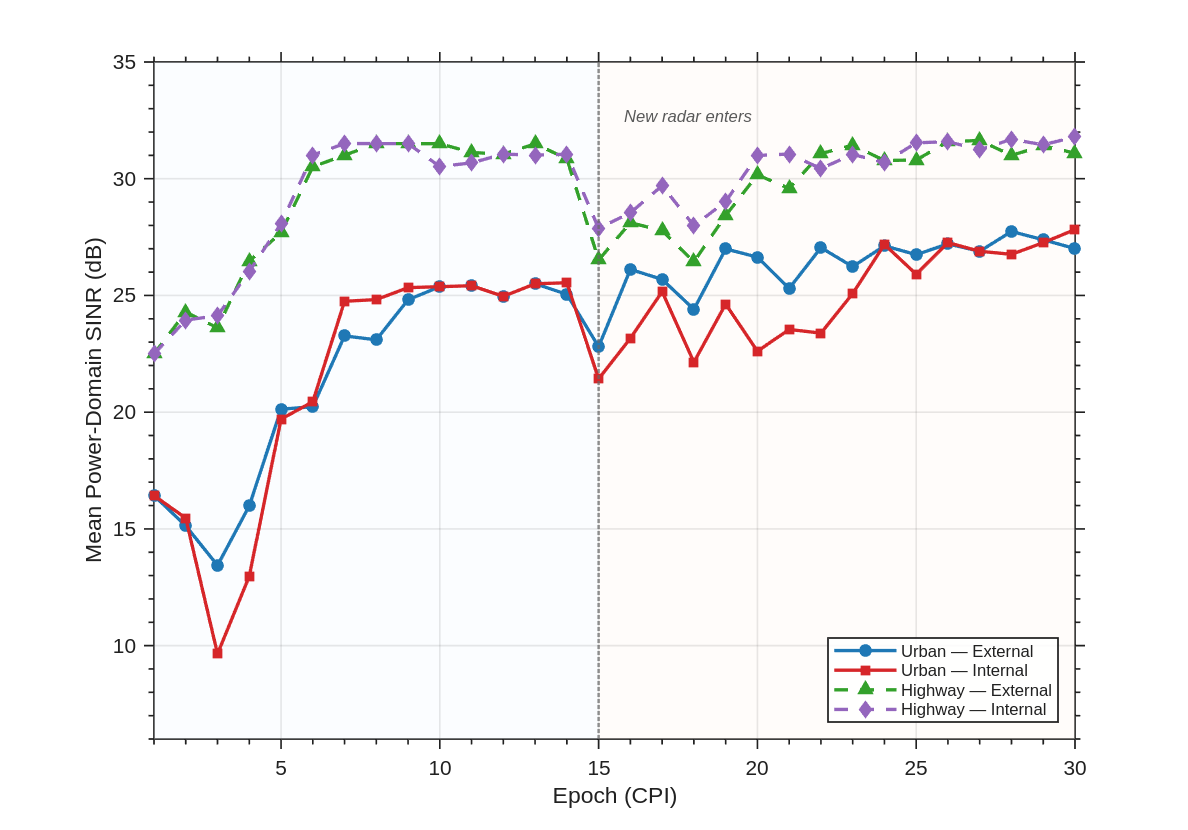}
    \caption{Mean SINR vs.\ epoch}
    \label{fig:dynamic_sinr}
\end{subfigure}
\caption{Dynamic scenario results. (a)~Interference-aware collision rate and (b)~mean SINR across radars for both urban and highway scenarios. The dashed vertical line marks new radar entry at epoch~15.}
\label{fig:dynamic_results}
\end{figure*}

\subsubsection{Timing analysis}
Each CPI spans $K \times T_\text{PRI} = 256 \times 29.99\,\mu\text{s} \approx 7.68$\,ms, so 30~epochs of learning consume approximately 230~ms of real time. Even in the fastest scenario (highway oncoming at 260~km/h closing speed over 200~m, encounter duration $\approx$2.8~s), the algorithm converges in less than 8\% of the encounter. For urban intersections (closing at $\sim$50~km/h over 150~m, encounter $\approx$10.8~s), convergence uses approximately 2\% of the available time. During this window, inter-vehicle positions change modestly relative to the 200~m interference range, although the interference graph is recomputed every CPI to reflect the current topology.

When new radars enter at epoch~15, the existing learned strategies are transiently perturbed but recover within 3--5~epochs. The new radars themselves converge from cold-start initialization, validating the framework's adaptability in dynamic traffic.

\section{Conclusion}
\label{sec:conclusion}
We have introduced a game-theoretic framework for proactive interference avoidance in automotive FMCW radar networks, generalizing previous frequency-hopping schemes to include time-domain shifts.  
By modeling joint time–frequency scheduling as a repeated anti-coordination game, we developed a decentralized no-regret learning algorithm that enables each radar to autonomously adapt chirp start times and subband selections without communication.
We proved sublinear external and internal regret bounds and established convergence of the empirical play to coarse correlated or correlated equilibria.
Numerical experiments demonstrated that leveraging both temporal and spectral degrees of freedom significantly reduces collision rates and improves SINR and range–Doppler fidelity compared to randomized time-frequency dithering and centralized Nash-based strategies. These results highlight the practicality and scalability of no-regret time–frequency scheduling for dense automotive radar scenarios. Future work will explore heterogeneous radar capabilities, more realistic channel models, hardware-in-the-loop validation on production platforms, formal detection performance analysis (probability of detection $P_d$ and false alarm rate $P_{fa}$) under the learned scheduling policies, and integration with large-scale traffic simulators such as SUMO for statistical evaluation across diverse traffic densities.

\section*{Acknowledgment}
% This work was conducted as part of a collaborative effort between New York University and NXP Semiconductors, and was partially supported by NXP Semiconductors. The authors report no conflicts of interest.
The authors used generative AI models gemini-2.5-pro and ChatGPT-5 for minor language editing. All technical content, analysis, and conclusions were developed by the authors.

\appendix
\begin{proof}[Proof of Theorem~\ref{thm:bandit_ext_regret}]
Let $w_\tau(\sigma):=\exp\{z^i_\tau(\sigma)\}$ and $W_\tau:=\sum_\sigma w_\tau(\sigma)$, so $\tilde p^i_\tau(\sigma)=w_\tau(\sigma)/W_\tau$.
By the update rule and setting $\eta\le\gamma/|\Sigma^i|$, we have $\eta  \widehat U^i_\tau(\sigma)\in[0,1]$ (since $\widehat U^i_\tau(\sigma)\le 1/p^i_\tau(\sigma_\tau)\le |\Sigma^i|/\gamma$).
For any $x\in[0,1]$, $e^x\le 1+x+x^2$. Hence,
\begin{align*}
W_{\tau+1}
&=\sum_\sigma w_\tau(\sigma)  \exp\{\eta  \widehat U^i_\tau(\sigma)\}
\\ 
& \le\
\sum_\sigma w_\tau(\sigma)\Big(1+\eta  \widehat U^i_\tau(\sigma)+\eta^2  \widehat U^i_\tau(\sigma)^2\Big)\\
&=
W_\tau\Big(1+\eta  \langle \tilde p^i_\tau,\widehat U^i_\tau\rangle+\eta^2  \sum_\sigma \tilde p^i_\tau(\sigma)  \widehat U^i_\tau(\sigma)^2\Big).
\end{align*}
Taking logs and using $\log(1+x)\le x$,
\begin{equation}\label{eq:logW_increase}
\begin{aligned}
   \log W_{\tau+1}-\log W_\tau
 \le\
 \eta \langle \tilde p^i_\tau,\widehat U^i_\tau\rangle
+\eta^2 \sum_\sigma \tilde p^i_\tau(\sigma)  \widehat U^i_\tau(\sigma)^2.
\end{aligned}
\end{equation}
On the other hand, for any fixed comparator $\sigma^\star$,
\[
\begin{aligned}
   \log \frac{w_{\mathcal{T}+1}(\sigma^\star)}{W_1}
 =\sum_{\tau=1}^{\mathcal{T}} \eta  \widehat U^i_\tau(\sigma^\star)
-\sum_{\tau=1}^{\mathcal{T}} (\log \frac{W_{\tau+1}}{W_\tau} ).
\end{aligned}
\]
Combining with \eqref{eq:logW_increase} and using $W_1=|\Sigma^i|$ and $w_{\mathcal{T}+1}(\sigma^\star)\le W_{\mathcal{T}+1}$ gives
\begin{equation}\label{eq:key_step_unexpected}
\begin{aligned}
 \sum_{\tau=1}^{\mathcal{T}} \langle \tilde p^i_\tau,\widehat U^i_\tau\rangle - \widehat U^i_\tau(\sigma^\star)
 \le 
\frac{\log |\Sigma^i|}{\eta}
 +
\eta  \sum_{\tau=1}^{\mathcal{T}}\sum_\sigma \tilde p^i_\tau(\sigma)  \widehat U^i_\tau(\sigma)^2.
\end{aligned}
\end{equation}
We now take expectations w.r.t.\ the learner’s sampling.
By standard importance-weighting, and the unbiasedness of utility estimation
$
\mathbb{E}_{\tau-1}\big[\widehat U^i_\tau(\sigma)\big]= \mathbb{E}_{\tau-1}\big[\bar U^i_\tau(\sigma) \big] = U^i_\tau,$ 
$\mathbb{E}_{\tau-1}\big[\langle \tilde p^i_\tau,\widehat U^i_\tau\rangle\big]=\mathbb{E}_{\tau-1}\big[\langle \tilde p^i_\tau, \bar U^i_\tau\rangle\big].
$
For the quadratic term, note that only the played arm contributes. Therefore, when $\gamma \neq 0, \gamma \leq \gamma_{\max}$:
\[
\sum_\sigma \tilde p^i_\tau(\sigma)  \widehat U^i_\tau(\sigma)^2
=\tilde p^i_\tau(\sigma_\tau)  \frac{\bar U^i_\tau(\sigma_\tau)^2}{p^i_\tau(\sigma_\tau)^2}
\ \le\
\frac{1}{p^i_\tau(\sigma_\tau)}\ \le\ \frac{|\Sigma^i|}{\gamma},
\]
since $U^i_\tau(\sigma_\tau)\le 1$, and $p^i_\tau\ge \gamma  \mathbf{1}/|\Sigma^i|$. Therefore
\[
\mathbb{E}\left[\sum_{\tau=1}^{\mathcal{T}}\sum_\sigma \tilde p^i_\tau(\sigma)  \widehat U^i_\tau(\sigma)^2\right]
\ \le\
\frac{\mathcal{T}  |\Sigma^i|}{\gamma}.
\]
Taking expectations in \eqref{eq:key_step_unexpected} and using
$\langle p^i_\tau,U^i_\tau\rangle=(1-\gamma)\langle \tilde p^i_\tau,U^i_\tau\rangle+\gamma\langle \mathbf{1}/|\Sigma^i|,U^i_\tau\rangle
\ge (1-\gamma)\langle \tilde p^i_\tau,U^i_\tau\rangle$ (hence
$\langle \tilde p^i_\tau,U^i_\tau\rangle\le \langle p^i_\tau,U^i_\tau\rangle+\gamma$),
we obtain
\begin{align*}
 \mathbb{E} \left[\sum_{\tau=1}^{\mathcal{T}} U^i_\tau(\sigma^\star) -  \langle p^i_\tau,U^i_\tau\rangle \right]
  \le 
\frac{\log |\Sigma^i|}{\eta}
+
\frac{\eta  \mathcal{T}  |\Sigma^i|}{\gamma}
+
\gamma \mathcal{T}.
\end{align*}
Rearranging gives \eqref{eq:regret_bound_main}. Optimizing the RHS over $\eta$ and then $\gamma$
yields the stated $\mathcal{O}\big(\mathcal{T}^{2/3}|\Sigma^i|^{1/3} (\log |\Sigma^i|)^{1/3}\big)$ bound. 

When $\gamma = 0$, 
\[
\begin{aligned}
    \sum_\sigma \tilde p^i_\tau(\sigma)  \widehat U^i_\tau(\sigma)^2
\leq (1 - \bar{U}^i_\tau )^2 + \sum_{\sigma^i \neq \sigma^i_\tau} \bar{U}^i_\tau (\sigma^i, \sigma^{-i}_\tau) \leq |\Sigma^i|. 
\end{aligned}
\]
Taking expectations in \eqref{eq:key_step_unexpected},  we obtain
\begin{equation*}
     \mathbb{E} \left[ \sum_{\tau=1}^{\mathcal{T}} U^i_\tau(\sigma^\star) -  \langle p^i_\tau,U^i_\tau\rangle \right]
  \le 
\frac{\log |\Sigma^i|}{\eta}
+
\eta  \mathcal{T}  |\Sigma^i| , 
\end{equation*}
Optimizing the RHS over $\eta$ and then $\gamma$
yields the stated $\mathcal{O}\big(\sqrt{\mathcal{T}|\Sigma^i|\log |\Sigma^i|}\big)$ bound. 
\end{proof}

\begin{proof}[Proof of Theorem~\ref{thm:swap_regret}]

We show the stated bound on the internal regret for a fixed radar $i$.
For each source action $\sigma\in\Sigma^i$, we define the \emph{row-wise} vector:
\begin{equation*}
    v^i_{\tau, \sigma} (\sigma^{\prime}) := p^i_\tau(\sigma) (\sigma^{\prime}) {U}^i_\tau(\sigma^{\prime}),
\end{equation*}
for every pure strategy $\sigma^i$ where $U^i_\tau(\cdot)\in[0,1]^{|\Sigma^i|}$ is the true utility vector at CPI $\tau$ for player $i$, and $p^i_\tau\in\Delta(\Sigma^i)$ is the mixed strategy used at that CPI.
Now, the algorithm has access only to bandit feedback via the estimator $\widehat U^i_\tau$ of $U^i_\tau$, according to which we can define estimator $\widehat v^i_\tau (\sigma^{\prime}):= p^i_\tau (\sigma) \hat{U} (\sigma^{\prime})$
We claim that $v^i_{\tau,\sigma}(\sigma')$ satisfies the following property: $\mathbb{E}_{\tau-1}[\widehat{v} (\sigma^{\prime}) ] = p^i_\tau U_\tau(\sigma^{\prime})$

For each fixed source $\sigma\in\Sigma^i$, Subroutine
~\ref{algo:sub2_int} runs a bandit external optimizer on the row game
with reward vectors $v^i_{\tau,\sigma}$, with utility estimator $\widehat{v}^i_{\tau,\sigma}$. 
Let $q^i_{\tau,\sigma}\in\Delta(\Sigma^i)$ denote the row strategy at CPI $\tau$ for source $\sigma$, applying
Theorem~\ref{thm:bandit_ext_regret} \emph{row-wise} yields: for any fixed $\sigma\in\Sigma^i$ and any comparator $\sigma'\in\Sigma^i$,
\begin{equation*}
\begin{aligned}
    \mathbb{E}\Bigg[
\sum_{\tau=1}^{\mathcal{T}}
v^i_{\tau,\sigma}(\sigma') -  \langle q^i_{\tau,\sigma}, v^i_{\tau,\sigma}\rangle
\Bigg] 
 & \le \frac{\log |\Sigma^i|}{\eta}
+
\frac{\eta |\Sigma^i|}{\gamma} \mathcal{T} + {\gamma} \mathcal{T}, \\
\text{ or when } \gamma \equiv 0 \qquad &  \le \frac{\log |\Sigma^i|}{\eta} + \eta \mathcal{T}|\Sigma^i| ,
\end{aligned}
\end{equation*}
where we have taken a constant learning rate $\eta_\tau\equiv\eta$ and exploration rate
$\gamma_\tau \equiv \min \{ \gamma_{\max}, \gamma\}$ for simplicity. 

Invoking the external-to-swap regret reduction of Blum \& Mansour~\cite[Theorem~5]{blum07externalinternal},  we obtain,
\begin{align*}
   & \quad \mathcal{R}^i_{\mathrm{int}}(\mathcal{T})   =  \max_{\phi} \sum_{\tau=1}^\mathcal{T} \langle  p^i_\tau (\cdot)  , U^i_\tau(\phi (\cdot) )  - U^i_\tau (\cdot) \rangle \\ 
    & = \max_{\phi}  \sum_{\tau=1}^\mathcal{T}  \langle  p^i_\tau (\cdot)  , Q^i_{\tau} U^i_\tau(\phi (\cdot) )  - Q^i_\tau U^i_\tau (\cdot) \rangle \\
    & = \max_{\phi}\sum_{\sigma \in \Sigma^i} \sum_{\tau=1}^{\mathcal{T}}  p^i_\tau(\sigma) \left( U^i_\tau (\phi(\sigma)) - \langle q_{\tau, \sigma}  (\cdot) , U^i_\tau(\cdot) \rangle \right) \\
    & \le \sum_{\sigma \in \Sigma^i} \max_{\sigma^{\prime}}\sum_{\tau=1}^{\mathcal{T}}  \left( v^i_\tau (\sigma^{\prime}) - \langle q_{\tau, \sigma} (\cdot) , v^i_{\tau, \sigma}(\cdot) \rangle  \right),
\end{align*}
where we eliminate the $\max_\phi$ by exchange the order of the inner product and summation over $\mathcal{T}$. Now each source pure strategy $\sigma$ corresponds to an external minimizer.

Therefore, by applying the row-wise result, we arrive at
\begin{equation*}
\begin{aligned}
\mathbb{E}\big[\mathcal{R}^i_{\mathrm{int}}(\mathcal{T})\big]
& \le
|\Sigma^i|\left(
\frac{\log |\Sigma^i| }{\eta}
+ \frac{\eta |\Sigma^i|}{\gamma} \mathcal{T} + {\gamma} \mathcal{T}
\right), \text{ or } \\
\mathbb{E}\big[\mathcal{R}^i_{\mathrm{int}}(\mathcal{T})\big] 
& \le
|\Sigma^i|\left(
\frac{\log |\Sigma^i| }{\eta}
+ \eta |\Sigma^i| \mathcal{T}
\right),
\end{aligned}
\end{equation*}
where we simply apply the fact that $p^i_\tau \le 1$, and use the result for the external regret bound. Now the results are almost immediate, e.g., treating $\gamma$ as a constant and optimizing over $\eta$,  we have $\eta = \sqrt{ \frac{ \gamma \log |\Sigma^i|}{(|\Sigma^i| \mathcal{T})}}$. Plugging the $\eta$ expression back in and optimizing $\gamma$ we obtain $\gamma = \left(\sqrt{\frac{|\Sigma^i| \log |\Sigma^i|}{\mathcal{T}} }\right)^{2/3}$, hence choosing $\eta = (\log |\Sigma^i|)^{2/3}|\Sigma^i|^{-1/3} \mathcal{T}^{-2/3}$ yields the stated internal regret bound. 

\end{proof}

\bibliographystyle{IEEEtran}
\bibliography{IEEEref}

\end{document}